\documentclass[twoside,leqno]{article}
\usepackage[letterpaper]{geometry}
\usepackage{ltexpprt}					
\usepackage{cite}                       
\usepackage{caption}                    
\usepackage{subcaption}                 
\usepackage{float}                      
\usepackage{chngcntr}
\usepackage{apptools}
\usepackage{cancel}
\usepackage{graphicx}                   
\usepackage{amsmath}                    
\usepackage{amssymb}                    
\usepackage{amsfonts}                   
\usepackage{mathrsfs}                   
\usepackage{url}                        
\usepackage{color}                      
\usepackage{hyperref}                   
\hypersetup{colorlinks=true,linkcolor=[rgb]{0,0,0},citecolor=[rgb]{0,0,0},urlcolor=[rgb]{0,0,0}}
\usepackage{enumitem}                   
\usepackage{thmtools}                   
\usepackage[parfill]{parskip}
\usepackage{comment}
\usepackage{microtype}

\allowdisplaybreaks

\graphicspath{{./Figs/}}

\newcommand{\floor}[1]{\left\lfloor #1\right\rfloor}

\newcommand{\ang}[1]{\left\langle #1 \right\rangle}
\newcommand{\inner}[2]{\left\langle #1, #2 \right\rangle}
\newcommand{\RE}{\mathbb{R}}

\newcommand{\eps}{\varepsilon}

\newcommand{\ST}{\,:\,}

\newcommand{\etal}{\textit{et al.}}
\newcommand{\bd}{\partial}
\newcommand{\Gradient}{\nabla}
\newcommand{\Hess}{\Gradient^2}
\newcommand{\DF}{d_P}
\newcommand{\DFSq}{D_P}

\newcommand{\sDF}{\widetilde{d}_P}
\newcommand{\sDFSq}{\widetilde{D}_P}
\newcommand{\patch}{\Pi}
\newcommand{\depth}{\Delta_{\patch}}
\newcommand{\shape}{\sigma}
\newcommand{\localF}{f}
\newcommand{\SP}{\kern+1pt}
\DeclareMathOperator*{\diam}{\mathrm{diam}}

\DeclareMathOperator*{\dist}{\mathrm{dist}}
\DeclareMathOperator*{\vol}{\mathrm{vol}}

\DeclareMathOperator*{\radius}{\mathrm{radius}}
\DeclareMathOperator*{\nn}{\mathrm{nn}}
\DeclareMathOperator*{\supp}{\mathrm{supp}}
\DeclareMathOperator*{\rep}{\mathrm{rep}}
\DeclareMathOperator*{\ch}{\mathrm{ch}}
\DeclareMathOperator*{\cell}{\mathrm{cell}}

\DeclareMathOperator{\interior}{int}

\begin{document}


\title{\Large Differentiable Approximations for Distance Queries}

\author{%
    Ahmed Abdelkader\thanks{Research conducted in part while at the University of Texas at Austin.}\\
		Google LLC \\
		Mountain View, CA 
		\and
    David M. Mount\\
		Department of Computer Science and \\
		Institute for Advanced Computer Studies \\
		University of Maryland, College Park 
}

\date{}

\maketitle

\begin{abstract}
The widespread use of gradient-based optimization has motivated the adaptation of various classical algorithms into differentiable solvers compatible with learning pipelines. In this paper, we investigate the enhancement of traditional geometric query problems such that the result consists of both the geometric function as well as its gradient. Specifically, we study the fundamental problem of distance queries against a set of points $P$ in $\mathbb{R}^d$, which also underlies various similarity measures for learning algorithms.

The main result of this paper is a multiplicative $(1+\varepsilon)$-approximation of the Euclidean distance to $P$ which is differentiable at all points in $\mathbb{R}^d \setminus P$ with asymptotically optimal bounds on the norms of its gradient and Hessian, from a data structure with storage and query time matching state-of-the-art results for approximate nearest-neighbor searching. The approximation is realized as a regularized distance through a partition-of-unity framework, which efficiently blends multiple local approximations, over a suitably defined covering of space, into a smooth global approximation. In order to obtain the local distance approximations in a manner that facilitates blending, we develop a new approximate Voronoi diagram based on a simple point-location data structure, simplifying away both the lifting transformation and ray shooting.
\end{abstract}



\newpage

{ 
 \hypersetup{linkcolor=black}
 \tableofcontents
}
\newpage

\section{Introduction} \label{sec:intro}

Gradient-based optimization is a cornerstone of end-to-end learning, enabling performant algorithms to be discovered from the data~\cite{LBH15, GBC16, ACG19}. Despite the demonstrated power of pure data-driven approaches, their performance and resource requirements can be dramatically improved by recognizing (classical) sub-problems with known solutions, e.g., from a suitable solver~\cite{CRB18, WDW19}. Before such solvers can be incorporated into the learning pipeline, they often need to be made \emph{differentiable}.

Several algorithms have already been adapted into differentiable variants for learning, starting with optimization algorithms~\cite{AmK17} for graph~\cite{ZDM23}, convex~\cite{AAB19}, non-linear~\cite{PFM22}, and general combinatorial~\cite{QSY22} problems. These new species of optimization algorithms quickly found applications in computer vision~\cite{GMA23}, computer graphics~\cite{LWF23}, physics simulations~\cite{FFR21}, and control~\cite{GHS21}. Beyond optimization, key algorithms in computer graphics and shape modeling were also adapted, including rendering~\cite{LoB14} and iso-surface extraction~\cite{GRL24}. Such differentiable solvers can significantly increase the generalizability of learned models as well as speed up convergence~\cite{ASA18}, however, their invocations in the \emph{inner loop} of these learning-based pipelines make their performance vitally important~\cite{HAL20}.

In parallel, the algorithm theory community studied a number of notions related to differentiability, including \emph{continuity}~\cite{KuY23a, KuY23b, KuY24} and \emph{stability}~\cite{MSV18, HKM21}. In contrast, we study the efficient approximation of distance functions as routinely used to define various optimization objectives, e.g., for learning~\cite{BIJ23} and robotics~\cite{CLH05}. Specifically, a point set $P$ in $\RE^d$, which might be discrete or continuous, naturally induces a \emph{distance function} or a scalar \emph{distance field}
\[
    d_P: x \mapsto \inf_{s\in P} \|x - s\|,
\]
where $\|\cdot\|$ denotes the Euclidean norm. Given a query point $x \in \RE^d$, the objective is to evaluate both $d_P(x)$ and $\Gradient d_P(x)$, where $\Gradient d_P(x)$ denotes the gradient of $d_P$ with respect to $x$.

Answering either type of query exactly presents a unique set of challenges, especially in dimensions $d \geq 4$. Practically speaking, computing the exact distance $d_P(x)$ takes close to linear time, where the diminishing $O(n^{1-1/d})$ speed-up requires indexing a spatial decomposition encoding the Voronoi diagram of $P$~\cite{LeW77}, which can have a complexity of $\Theta(n^{\floor{d/2}})$~\cite{AMNSW98, HIM12}. Second, $\Gradient d_P(x)$ is not defined when $x$ is a critical point of $d_P$, i.e., when the $d_P(x)$ is realized by more than one point in $P$~\cite{BaE17}. Although generalized gradients can be employed to define a unique \emph{driver} at each point $x$, the resulting \emph{flow} field is itself non-smooth and its structure, known as the \emph{flow complex}, is closely related to that of the Voronoi diagram making it prohibitively large to store and index~\cite{BDG08}.

To address those challenges, we take inspiration from established works on \emph{regularized distance functions}~\cite{Ste70, Fra79}, where the objective is to define an approximating function that is everywhere differentiable.  Beyond their key role in extension theorems, and more broadly harmonic analysis, regularized distances were also applied in the context of geodesic distances~\cite{CLP20, EGS23}.

The closest work to ours is the regularized distance of Stein~\cite{Ste70}. His approach is based on a partition-of-unity over a Whitney decomposition, where the size of each cell is proportional to its distance from $P$. However, Stein's distance can be smaller than the true distance, which deviates from common definitions in computational geometry. Moreover, Stein's distance is bounded above and below by \emph{constant factors}, not accommodating a controllable approximation parameter $\eps$. (This limitation was later addressed by Fraenkel~\cite{Fra79}.) In addition, the Whitney decomposition employed by Stein~\cite{Ste70} is built using hypercubes, which cannot be used to obtain the best storage bounds.

In contrast to the works cited above, our paper concerns computationally-efficient multiplicative distance approximations achieving state-of-the-art storage bounds. While we are not aware of any work on gradient queries, the case for distance queries bears similarity to the well-studied approximate nearest-neighbor (ANN) search problem~\cite{AMM09a, HIM12, HaK15}. A natural approach to approximate $d_P(x)$ is to utilize an ``off-the-shelf'' ANN data structure to find the closest point $p \in P$ and then return $\|q-p\|$. Given an approximation parameter $\eps > 0$, such data structures return an \emph{$\eps$-approximate nearest neighbor} ($\eps$-ANN), or \emph{witness}, which is a point that is within a factor of $1+\eps$ of the true closest distance~\cite{AbM23}. Hence, we are guaranteed that $\|q - p\| \leq (1 + \eps)\DF(q)$.

An unfortunate consequence of this approach is that the resulting distance field encoded by traditional $\eps$-ANN data necessarily structures suffers from \emph{jump discontinuities}, rendering it unfit as a differentiable approximation. We proved the following result in an earlier paper.

\begin{lemma}[Abdelkader and Mount~\cite{AbM23}] \label{impossibility.lem}
If a witness-based distance function for a finite point set $P \subset \RE^d$ is inexact at even one point, it cannot be both continuous and provide a finite bound on relative errors.
\end{lemma}

It is not hard to see why. Intuitively, if the result is inexact at any one query point, then by moving the query point along a line segment to its true nearest neighbor, it can be shown that the witness must jump to the true nearest neighbor at some point along this segment, and a discontinuity must arise at this jump. In this same paper, we also presented a data structure for continuously approximating the distance to a polytope's boundary from points in its interior. This data structure achieved state-of-the-art storage bounds through the application of \emph{anisotropic} covering elements derived from Macbeath regions~\cite{AFM17a}. However, that data structure is based on a single type of smoothing elements (ellipsoids, in particular) to derive progressively-finer coverings of the entire polytope. These ellipsoids were arranged into a hierarchy through the use of a rooted directed-acyclic-graph. In contrast, efficient data structures for geometric query problems often require more flexibility in their design. In particular, they require a greater variety of covering elements combined into more sophisticated arrangements. Furthermore, since the data structure of~\cite{AbM23} only provided bounds on the absolute error, rather than relative error, the corresponding bound on the norm of the Hessian followed in a rather straightforward manner. In contrast, the results of this paper involve a more intricate interplay between efficiency and smoothness for geometric approximations. (See Section~\ref{sec:norm_bounds} for further details.)

In this sense, our work presents a smoothed geometric data structure that combines many of the benefits of prior works~\cite{Ste70,Fra79,AbM23}, and presents a more mature framework for designing smooth data structures enabling efficient differentiable approximations, starting with the fundamental problem of distance to a set of points in Euclidean space. Summarizing, we study the following problem:

\begin{description}
\item[Differentiable Distance-Query Problem:]
Given a point set $P \subset \RE^d$ and an approximation parameter $\eps > 0$, implement a queryable continuous approximation $\widetilde{d}_P: \RE^d \to \RE$ satisfying $d_P(x) \leq \widetilde{d}_P(x) \leq (1+\eps)\cdot d_P(x)$. For an input query $x$, the output is both the distance approximation $\widetilde{d}_P(x)$ as well as its gradient $\Gradient \widetilde{d}_P(x)$.
\end{description}

We will present a solution to this problem when the set in question is a set $P$ of $n$ points in $\RE^d$ (where $d$ is a fixed constant). The space and query time complexities are essentially the same as those of the best known data structures for answering $\eps$-ANN queries. In addition, we obtain optimal bounds~\cite{Ste70,Fra79} on the rate of change of the distance function encoded by the data structure, in terms of the norms of its gradient and Hessian. Throughout, we assume that the dimension $d$ is a constant independent of $n$ and $\eps$.

\begin{theorem}\label{thm:puann}
Given a set $P$ of $n$ points in $\RE^d$ and an approximation parameter $0 < \eps \leq 1/2$, there exists a smooth function $\sDF$ satisfying $\DF(x) \leq \sDF(x) \leq (1+\eps)\cdot\DF(x)$ for all $x \in \RE^d$, which can be evaluated at any point along with its gradient in $O\big( \log \frac{n}{\eps} \big)$ time from a data structure using $O\big(n / \eps^{d/2} \big)$ space.%
\footnote{Formally, the gradient of the Euclidean distance is undefined whenever $x \in P$. Our algorithm actually computes the squared Euclidean distance, and this function has well-defined gradients at all points.}
Further, the norms of the gradient and Hessian of $\sDF$ are asymptotically optimal, satisfying
\[
    \|\Gradient \sDF(x)\| 
        ~ = ~ O(1)
        \qquad \text{and} \qquad
    \|\Hess \sDF(x)\| 
        ~ = ~ O\bigg(\frac{1}{\eps\cdot\DF(x)}\bigg).
\]
\end{theorem}

Our approach involves an adaptation of a well-known technique for generating smooth approximations, called the \emph{partition of unity}~\cite{Lee03}. The idea involves smoothly blending a bounded number of local continuous approximations into a continuous whole. The challenge in the context of data structure design is how to do this without compromising the efficiency of the underlying data structure. The blending process involves collecting and combining nearby local approximations. We will employ a combination of elements including a quadtree-like decomposition of space and a hierarchical cover by ellipsoids, which we call an Ellipsoidal (approximate) Voronoi Diagram (or EVD). We develop the EVD data structure by leveraging recent progress on approximate nearest neighbor searching~\cite{AFM17a}, while simplifying away both the lifting transformation and ray shooting. In particular, we utilize an efficient geometric construction derived from Macbeath regions. These constructions can be applied in the non-continuous context and may be of independent interest (see Theorem~\ref{thm:evd_ws}).

The remainder of the paper is organized as follows. We start with an overview and self-contained development of the partition-of-unity technique in Section~\ref{sec:pou}, charting a pathway to establish Theorem~\ref{thm:puann} with the desired optimal bounds on the norms of both the gradient and the Hessian. We then switch gears to distance approximation, where we study the \emph{well-separated} case in Section~\ref{sec:approx-voronoi}. To solve this fundamental restriction of the general problem, we develop a novel ellipsoidal cover based on a distance-based generalization of the Macbeath regions routinely defined with respect to a convex body. We proceed to design of a new $\eps$-ANN data structure based on this new ellipsoidal covering, the EVD data structure, which we then utilize in a standard modular way to obtain an efficient data structure for the general problem. Finally, we describe the adaptations needed to enable blending the local approximations, fulfilling the plan of Section~\ref{sec:pou}. Throughout, a number of the more tedious technical details have been deferred to the appendix.

\section{Continuity Through Local Blending} \label{sec:pou}

As mentioned above, our approach is based on blending a bounded number of local approximations to the distance function. We begin by introducing the partition-of-unity approach (see, e.g., Lee~\cite{Lee03}), and then we will explain how to adapt it to the context of distance computation.

\subsection{Smoothing and Partition of Unity.}

Given an open cover $\Pi = \{\patch_i\}_{i \in I}$ of a domain $X$, a \emph{partition of unity subordinate to $\Pi$} is a collection of smooth functions $\{\phi_i : X \rightarrow \RE\}_{i \in I}$ such that:
\begin{enumerate}
\item[(i)] $0 \leq \phi_i(x) \leq 1$, for all $i \in I$ and $x \in X$,

\item[(ii)] $\supp(\phi_i) \subset \patch_i$ for $i \in I$,

\item[(iii)] the set of supports $\{\supp(\phi_i)\}_{i \in I}$ is locally finite, and

\item[(iv)] $\sum_{i \in I} \phi_i(x) = 1$, for all $x \in X$.
\end{enumerate}
The elements of $\Pi$ are called \emph{patches}.

How is this adapted to the context of distance approximation over a point set $P$? As common in the design of geometric data structures, space is decomposed into regions or cells. Each cell is made small enough that it is possible to compute a local approximation to the distance function, typically involving a single witness from $P$.

\begin{figure}[htbp]
  \centerline{\includegraphics[scale=0.4]{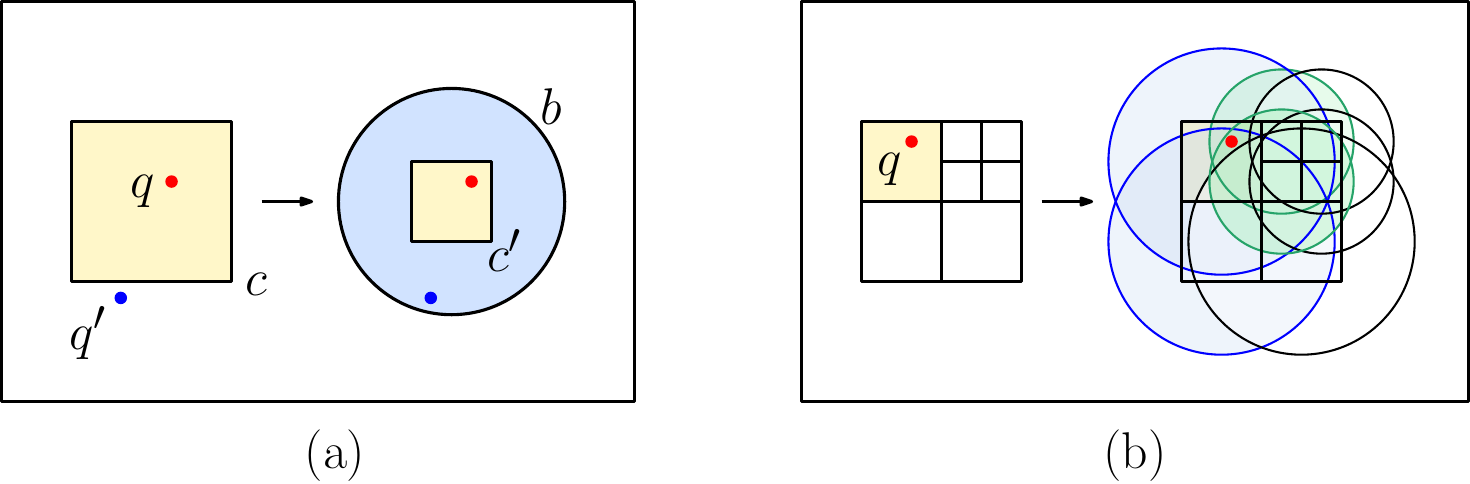}}
  \caption{Blending in the context of a quadtree-based AVD: (a) A leaf cell $c$ is effectively replaced by a patch, the ball $b$, which can be used to answer queries beyond its defining cell such as $q'$, and a smaller cell $c'$ is used to ensure such approximations are valid. (b) Instead of using a single cell to answer the query, local approximations from multiple patches are blended, e.g., two blue and two green.} 
  \label{AVD_blending.fig}
\end{figure}

Our approach to blending is most easily understood in the context of Har-Peled's quadtree-based construction~\cite{Har01, Har11a}. Ignoring the complications due to compression, each cell of the decomposition is a hypercube in $\RE^d$. For each leaf cell, we generate a patch in the form of a Euclidean ball centered about the cell and containing a constant-fraction expansion of the cell (see Figure~\ref{AVD_blending.fig}(a)). Patches from nearby cells will naturally overlap each other, but by controlling the sizes of nearby cells, we can guarantee that no point lies in more than a constant number of patches. Let $\depth$ denote this constant. We engineer the cell sizes so that the nearest-neighbor representative that works for query points within the cell, works as well for all query points within the larger associated patch. Now, all that remains is to blend the local approximations from these various patches, based on the distance function for each cell induced by its representative (see Figure~\ref{AVD_blending.fig}(b)).

Let us delve a little deeper into the details of how this blending is done. First, it will be more convenient to work with squared Euclidean distances $\DFSq = \DF^2$~\cite{Sie99}. (It can be shown that this only affects the approximation parameter $\eps$ by a constant factor.%
\footnote{When distances are squared, the relative error changes from $1+\eps$ to $(1+\eps)^2$. Assuming that $\eps \leq 1$, we have $(1+\eps)^2 = 1 + 2\eps + \eps^2 \leq 1 + 3 \eps$. Therefore, by setting $\eps' = \eps/3$, an $\eps'$-ANN for squared distances is an $\eps$-ANN for standard distances. This adjustment only alters the constant factors hidden in the asymptotic notation.}%
) As above, let $\{\patch_i\}_{i \in I}$ denote the set of patches, and for each $\patch_i$, let us assume that we have a \emph{local approximation} $\localF_i$, such that for all $x \in \patch_i$, $\DFSq(x) \leq \localF_i(x) \leq (1+\eps)\DFSq(x)$. This function can be implemented in many ways. We will assume throughout that each patch $\patch_i$ is associated with a \emph{representative point} of $P$, denoted $\rep(\patch_i)$, and $\localF_i(x)$ is defined to be the squared distance from $x$ to this point,
\begin{equation} \label{eq:v_i}
    \localF_i(x) ~ = ~ \|x - \rep(\patch_i)\|^2.
\end{equation}

As mentioned above, $\phi_i$ is a smooth function whose support is contained within $\patch_i$, and together these functions sum to unity for any point in space. To enforce the unity condition, the functions $\{\phi_i\}$ are defined by normalizing a set of smooth non-negative \emph{weight functions} $\{\psi_i\}$. A global approximation to the squared distance can be obtained from these local approximations as:
\begin{equation}\label{eq:pu}
    \sDFSq(x) 
        ~ = ~ \sum_{i \in I} \phi_i(x) \cdot \localF_i(x), 
                \quad \text{~~with~~} 
    \phi_i(x) 
        ~ = ~ \frac{\psi_i(x)}{\Psi(x)} \text{~~and~~} 
    \Psi(x) 
        ~ = ~ \sum_{i \in I} \psi_i(x).
\end{equation}
Although these sums are ostensibly taken over all the patches, they need only be applied to the at most $\depth$ patches that overlap $x$. The desired smooth distance approximation is then defined as
\begin{equation} \label{sDF.eq}
    \sDF 
        ~ = ~ (\sDFSq)^{1/2}, 
                \quad \text{~~with~~} 
    \|\Gradient \sDF\| 
        ~ = ~ \frac{\|\Gradient \sDFSq\|}{2\sDF}.
\end{equation}

The weight functions $\psi_i$ associated with each patch $\patch_i$ depend on the patch's shape. Eventually, we will consider three different patch shapes, but for the sake of simplicity, let us start with the simplest of these, namely Euclidean balls. Let $c^{[b]}$ and $r$ denote the center and radius of such a ball, respectively. First, for $x \in \RE^d$, define
\[
    \shape^{[b]}(x) 
        ~ = ~ \frac{1}{r^2}\|x - c^{[b]}\|^2.
\]
To restrict the support of the weight function within a ball of radius $r$, we compose $\shape^{[b]}$ with the \emph{bump function}
\begin{equation} \label{eq:bump}
    \mu(s) 
        ~ = ~ \begin{cases}
                \exp\left(\dfrac{1}{s^2 - 1}\right),   & |s| < 1,\\
                0,                                     & \text{otherwise.}
            \end{cases}
\end{equation}
Observe that the weight is highest at the center with $\mu(\shape^{[b]}) = \mu(0) = e^{-1}$, and smoothly decays towards the boundary where $\mu(\shape^{[b]}) = \mu(1) = 0$. Since $\mu \in C_c^\infty(\RE)$ and is non-analytic with vanishing derivatives for $|s| = 1$~\cite{Sie99}, the support is restricted as desired.

In order to establish a positive lower bound $c_\Psi$ on $\Psi(x)$, we will need to ensure sufficient overlap between patches. Using the weight functions $\{\psi_i\}$ and corresponding local approximations $\{\localF_i\}$ associated with the patches $\{\patch_i\}$, as in Eq.~\eqref{eq:v_i}, the smooth distance approximation $\sDFSq$ can be readily computed per Eq.~\eqref{eq:pu}.

\subsection{Designing Covers for Storage Efficiency.} \label{sec:finer}

The simple ball shape function $\shape^{[b]}$ is adequate for data structures based on simple isotropic decompositions (such as grids and non-compressed quadtrees). However, in order to match the best-known space bounds for $\eps$-ANN queries, we will need to develop a cover based on more sophisticated shapes. In particular, we use rings (that is, the set-theoretic difference of two Euclidean balls), ellipsoids, and intersections of these shapes.

The shapes we use are all centrally symmetric, and we use $c^\circ$ to denote the center of the shape. The other relevant parameters are a radius $r$ for balls, a positive-definite matrix $M$ for ellipsoids, and a pair of inner and outer radii $a$ and $b$, respectively, for rings. For a given desired shape, we set the parameters defining the shape functions $\shape^\circ$ such that the interior points of the shape are mapped to $(-1, 1)$ and its boundary points are mapped to $\{-1, 1\}$. We distinguish the type of the shape in the superscript, where the parameters can be associated unambiguously to the corresponding patch $\patch_i$ in the subscript. Generically, we obtain the following functional forms: $\shape^{[b]}$ for balls, $\shape^{[e]}$ for ellipsoids, and $\shape^{[r]}$ for rings.
\begin{align}
    \text{Ball:}~~~
        & \shape^{[b]}(x) 
        ~ = ~ \frac{1}{r^2}\left\| x - c^{[b]} \right\|^2 \label{eq:shape-b} \\
    \text{Ellipsoid:}~~~
        & \shape^{[e]}(x) 
        ~ = ~ \left( x - c^{[e]} \right)^{\intercal}M\left( x - c^{[e]} \right) \label{eq:shape-e} \\
    \text{Ring:}~~~
        & \shape^{[r]}(x) 
        ~ = ~ \dfrac{\left\| x - c^{[r]} \right\|^2 - (b^2+a^2)/2}{(b^2-a^2)/2}. \label{eq:shape-r}
\end{align}

In the proposed data structure, all ellipsoids are restricted within a ball, and some ellipsoids are additionally restricted within a ring. Restricting the support of weight functions to the intersection of shapes can be conveniently achieved by multiplication. For the data structure we develop in this paper, the weight function $\psi_i$ for patch $\patch_i$ takes one of the two forms below. (See Figure~\ref{fig:composition} for an example.)
\begin{equation}\label{eq:psi_i}
    \psi_i(x) 
        ~ = ~ \mu(\shape^{[b]}_i(x))\cdot\mu(\shape^{[e]}_i(x))
    \quad \text{ or } \quad 
    \psi_i(x) 
        ~ = ~ \mu(\shape^{[b]}_i(x))\cdot\mu(\shape^{[r]}_i(x))\cdot\mu(\shape^{[e]}_i(x)).
\end{equation}

\begin{figure}[htbp]
  \centerline{\includegraphics[scale=0.6]{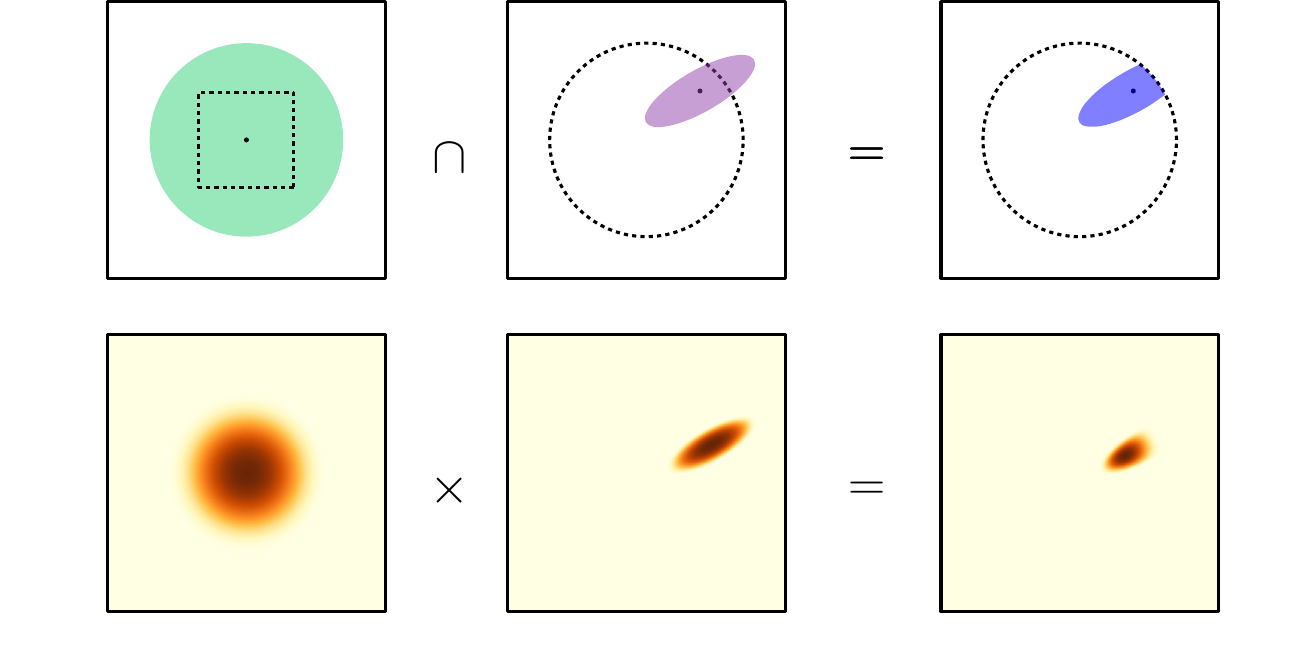}}
  \caption{Shapes (top) and corresponding weight functions (bottom). Restricting support to the intersection (right) by multiplying mollified functional forms for the ball (left) and ellipsoid (middle).}
\label{fig:composition}
\end{figure}

\subsection{Compatible Shape Functions.} \label{sec:shape}

As outlined in the previous section, the patches we employ for the partition-of-unity construction will be realized using the product of a collection of \emph{shape functions} with support in the interior of the patch. To streamline the derivation, we introduce a convenient layer of abstraction before restricting attention to a specific family of shape functions.

Formally, each patch $\Pi_i$ is constructed as the intersection of a set of \emph{component shapes} $\{\mathcal{S}_i^{(k)}\}$, where each component shape is associated with a shape function $\shape_i^{(k)}$. Letting $k_i$ denote the number of component shapes used for the patch $\Pi_i$, we have
\[
    \Pi_i 
        ~ = ~ \bigcap_{k \in [k_i]} \mathcal{S}_i^{(k)}
\]

While the only cases encountered in this paper involve $k_i \in \{2, 3\}$ (recall Eq.~\eqref{eq:psi_i}), it will be convenient to generalize this by defining a uniform constant upper bound $\kappa_\Pi$, such that $k_i \leq \kappa_\Pi$ for all $i \in I$. This allows us to express any weight function as
\begin{equation}\label{eq:generalized_weight}
    \psi_i(x) 
        ~ = ~ \prod_{k \in [k_i]} \mu\left(\shape_i^{(k)}\right), \quad\text{~~with~~} 
    k_i 
        ~ \leq ~ \kappa_\Pi.
\end{equation}

Furthermore, we require all shape functions to satisfy the following definition. The first condition constrains the shape function's values, and the second condition bounds the growth rates of the function's gradient and Hessian.

\begin{Definition}[Compatible Shape Function] \label{def:compat-shape}
Given a distance-query problem for a set $P$ and error bound $\eps > 0$, the shape function $\shape$ associated with a component $\mathcal{S}$ is \emph{compatible} if and only if there exist constants $c_1$ and $c_2$ such that for all $x \in S$:
\begin{enumerate}
\item[$(i)$] $\shape(x) \in (-1,+1), \forall x \in \interior(S)$, and $\shape(x) \in \{-1,+1\}, \forall x \in \bd S$,

\item[$(ii)$] $\| \Gradient \shape(x) \| \leq \dfrac{c_1}{\eps\cdot \DF(x)}$ and $\| \Gradient^2 \shape(x) \| \leq \dfrac{c_2}{\eps^2\cdot \DF^2(x)}$.
\end{enumerate}
\end{Definition}

As we show later, all the shape functions employed in this paper are indeed compatible. Nonetheless, this added layer of abstraction allows us to derive some of the key results without dealing with the idiosyncrasies of each component shape. (Note that similar compatibility abstractions can be defined for other geometric functions with their corresponding upper bounds.)

\subsection{Bounding the Norms of the Gradient and Hessian.} \label{sec:norm_bounds}

The utility of the abstract patch model given in Eq.~\eqref{eq:generalized_weight} comes to light in this section, where we obtain an explicit and compact expression for the gradient of the smooth distance approximation from Eq.~\eqref{eq:pu}, which can be easily evaluated in the course of answering distance-function queries. Moreover, we obtain upper bounds on the norms of both the gradient and the Hessian, whenever the shape functions are compatible. The proof involves a fairly straightforward application of calculus and can be found in Section~\ref{sec:sDFGradient}.

\begin{restatable}{lemma}{sDFGradient}\label{sDF_gradient.lem}
The gradient of smooth distance approximations $\sDFSq$ blended over any set of patches $\{\Pi_i\}_{i \in I}$ realized by any set of compatible shape functions $\{\shape_i^{(k)}\}$ takes the form
\begin{equation} \label{eq:grad_sDF_exp}
    \Gradient \sDFSq(x)
        ~ = ~ 2 \sum_{i \in I} \phi_i(x)\cdot(x-\rep(\patch_i)) + \sum_{i \in I}\sum_{k \in k_i} \zeta_i^{(k)}(x)\cdot\Gradient\shape_i^{(k)}(x),
\end{equation}
where the coefficients are
\begin{equation} \label{eq:shorthand_coefficient}
    \zeta_i^{(k)}(x) 
        ~ = ~ \frac{2}{\Psi(x)}\cdot\frac{\psi_i(x)\cdot \shape_i^{(k)}(x)}{(\shape_i^{(k)}(x)^2-1)^2}\cdot\left(\sDFSq(x) - \localF_i(x)\right),
\end{equation}
with an upper bound $|\zeta_i^{(k)}(x)| = O\left(\eps\cdot\DF^2(x)\right)$, for all $i$ and $k$. In addition, for all points $x \in \RE^d$, the magnitude of the gradients are bounded as
\[
    \|\Gradient \sDFSq(x)\| 
        ~ = ~ O\left(\DF(x)\right) \text{~~ and ~~} 
    \|\Gradient \zeta_i^{(k)}(x) \| 
        ~ = ~ O(\DF(x)).
\]
\end{restatable}

The following lemma bounds the norm of the smooth approximation's Hessian. The proof can be found in Section~\ref{sec:sDFHessianBound}.

\begin{restatable}{lemma}{sDFHessianBound}\label{sDF_Hessian_bound.lem}
The Hessian of smooth distance approximations $\sDFSq$ blended over any set of patches $\{\Pi_i\}_{i \in I}$ realized by any set of compatible shape functions $\{\shape_i^{(k)}\}$ is upper-bounded by
\[
    \|\Gradient^2 \sDFSq(x)\| 
        ~ = ~ O(1 / \eps).
\]
\end{restatable}

Using Lemmas~\ref{sDF_gradient.lem} and~\ref{sDF_Hessian_bound.lem}, we obtain the following bounds on the square-root of the smooth distance approximation, which can be readily used as a valid (and smooth) approximation to the original distance function. Recalling Eq.~\eqref{sDF.eq}, we have
\[
    \sDF 
        ~ = ~ (\sDFSq)^{1/2}, \text{~~with~~} 
    \|\Gradient \sDF\| 
        ~ = ~ \frac{\|\Gradient \sDFSq\|}{2\sDF} 
        ~ = ~ O(1).
\]

In addition, to obtain a bound on $\Hess \sDF(x)$, we recall that
\[
    \sDF 
        ~ = ~ \sDFSq^{1/2} 
    \implies 
    \Gradient \sDF 
        ~ = ~ \frac{\Gradient \sDFSq}{2\sDFSq^{1/2}} 
    \implies 
    \Hess \sDF 
        ~ = ~ \frac{\Hess \sDFSq}{2\sDFSq^{1/2}} - \frac{1}{4\sDFSq^{3/2}}\cdot\left(\Gradient \sDFSq\right)\left(\Gradient \sDFSq\right)^\intercal.
\]
Hence, we can bound the norm of the Hessian as
\[
    \|\Hess \sDF(x) \| 
        ~ \leq ~ \frac{O(1/\eps)}{2\cdot\DF(x)} + \frac{1}{4\cdot\DF(x)^3} O(\DF(x)^2) 
        ~ =    ~ O\bigg(\frac{1}{\eps\cdot\DF(x)}\bigg).
\]
Stein~\cite{Ste70} and Fraenkel~\cite{Fra79} established that these bounds on the norms of the gradient and Hessian are asymptotically optimal.

\section{Well-Separated ANN: Ellipsoidal AVD} \label{sec:approx-voronoi}

In order to apply the results of the previous section to the differential distance-query problem, we will establish the existence of a concise and efficiently searchable partition of unity system for this problem. In this section, we begin with the easier case when the query region is well separated from the points of $P$
\footnote{A closely related construction was recently employed for kernel density estimation~\cite{CKW24}.}
, and the general case will be presented in Section~\ref{sec:quadtree_evd}.

Recall that $P$ denotes the set of points for the distance-query problem, and consider a query region $w \subseteq \RE^d$, which we will take to be a hypercube in $\RE^d$. For non-empty subsets $X, Y \subset \RE^d$, let $\diam(X)$ denote the diameter of $X$, and $\dist(X, Y)$ denote the minimum Euclidean distance from any $x \in X$ to any $y \in Y$. Given $\gamma > 1$, we say that $X$ is \emph{concentrically $\gamma$-separated} from $Y$ if $X$ can be enclosed in a $d$-dimensional ball $b$ such that $Y$ lies entirely outside $\gamma b$. We refer to $X$ as the \emph{inner set} and $Y$ as the \emph{outer set}. When we apply these concepts, one of these sets will be a subset of $P$ while the other will be a cell $w$ of the data structure (see Figure~\ref{separation_lite.fig}).

\begin{figure}[htbp]
  \centerline{\includegraphics[scale=0.40]{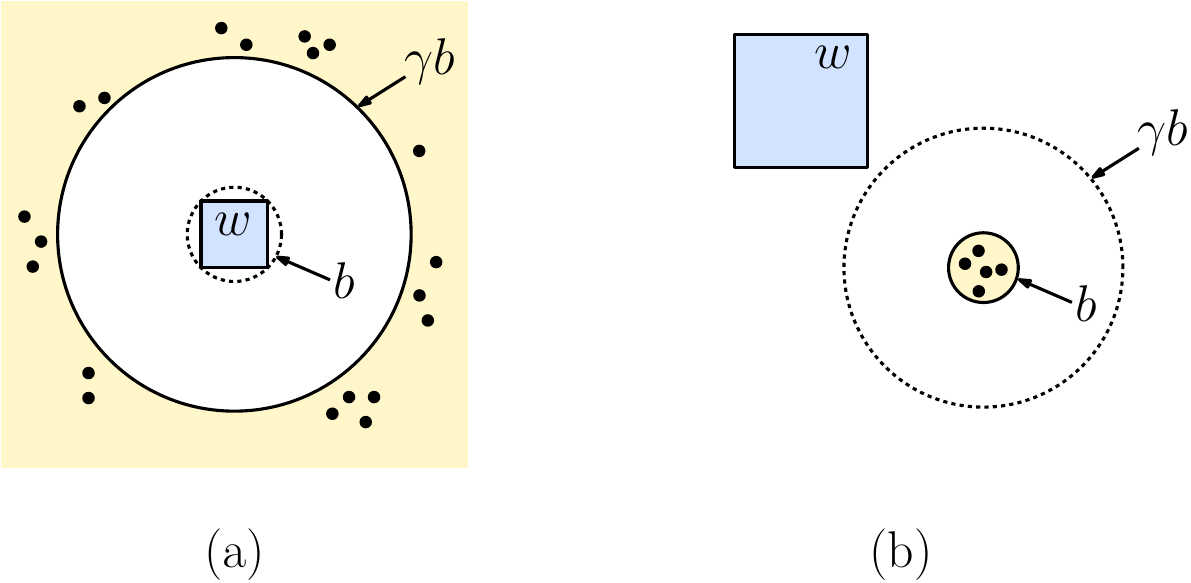}}
  \caption{Concentric $\gamma$-separation, where (a) $w$ is inner, $P$ is outer and (b) $w$ is outer, $P$ is inner.}
  \label{separation_lite.fig}
\end{figure}

In this section, we start by defining an \emph{expanded} Voronoi decomposition, highlighting its relevance for distance approximation when sufficient separation can be ensured. Then, we proceed to define an efficient ellipsoidal cover to approximate the Voronoi diagram. Finally, we analyze the resulting data structure, establishing its correctness and performance bounds.

\subsection{Approximation and Separation.} \label{sec:sep-approx}
Given a discrete point set $P \subset \RE^d$, recall that the \emph{Voronoi cell} of a point $p \in P$, denoted $V(p)$, is the set of points in $\RE^d$ having $p$ as a closest site. Given $\delta \geq 0$, we define the \emph{$\delta$-expanded Voronoi cell} of $p$ to be
\[
    V_{\delta}(p)
	~ = ~ \{ x \in \RE^d : \|x - p\|^2 ~\leq~ \|x - p'\|^2 + \delta^2, ~\forall p' \in P \}
\]
(see Figure~\ref{expanded-voronoi.fig}(a)). It is an easy exercise to show that $V_{\delta}(p)$ is the convex polytope that results by translating each $(p, p')$ bisector by a distance of $\delta^2/2\|p'-p\|$ away from $p$ (see Figure~\ref{expanded-voronoi.fig}(b)).

\begin{figure}[htbp]
  \centerline{\includegraphics[scale=0.40]{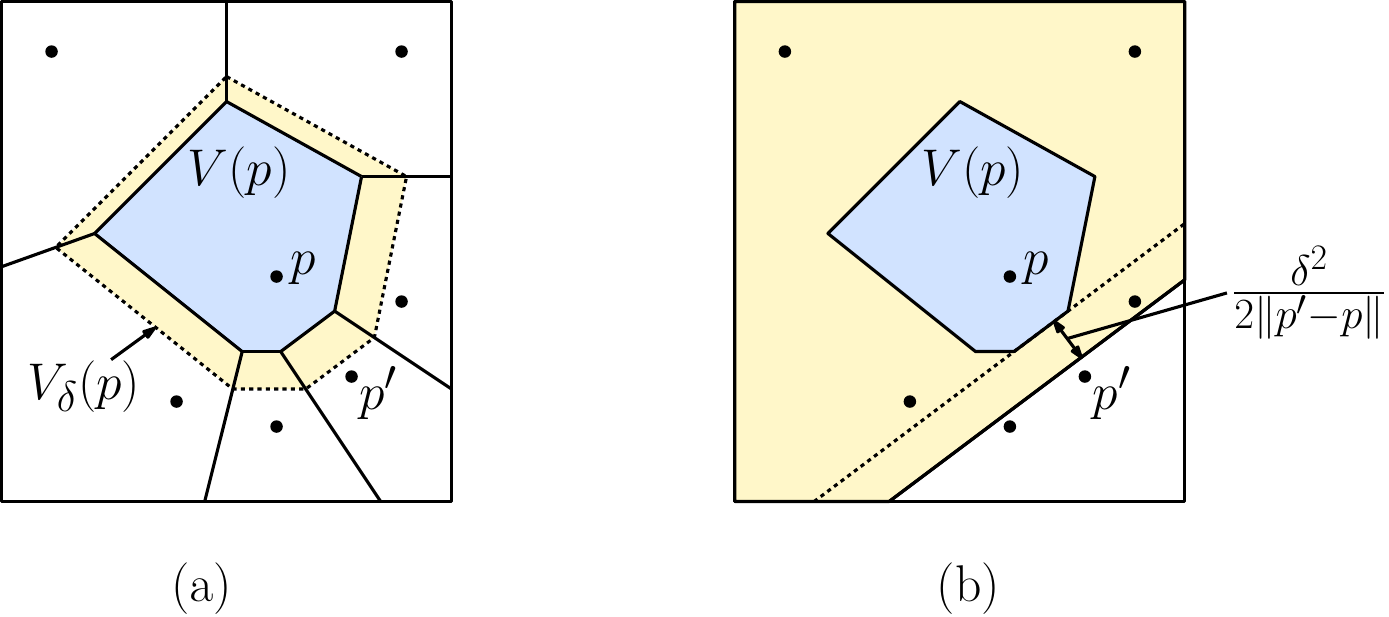}}
  \caption{(a) $\delta$-expanded Voronoi cell and (b) $\delta$-shifted bisector halfspace.} \label{expanded-voronoi.fig}
\end{figure}

While the expanded Voronoi cell is defined in terms of an additive error applied to the squared distance, it can still be used for approximate nearest-neighbor searching, provided that the ratio between the expansion factor and the nearest-neighbor distance is sufficiently small relative to $\sqrt{\eps}$.

\begin{lemma} \label{exp-vor.lem}
Consider a finite set $P \subset \RE^d$, $0 < \eps \leq 1/2$, any site $p \in P$, and $\delta > 0$. Given $q \in \RE^d$, if $q \in V_{\delta}(p)$ where $\delta \leq \sqrt{\eps} \|q - p\|$, then $p$ is an $\eps$-ANN of $q$.
\end{lemma}

\begin{proof}
If $q \in V_{\delta}(p)$ then for any $p' \in P$, we have
\[
    \|q - p\|^2 
        ~ \leq ~ \|q - p'\|^2 + \delta^2
        ~ \leq ~ \|q - p'\|^2 + \eps \|q - p\|^2,
\]
and hence $\|q - p\|^2 \leq \|q - p'\|^2/(1-\eps)$. It is easy to verify that for all $\eps \leq 1/2$, $1/(1-\eps) \leq (1 + \eps)^2$, and therefore $\|q - p\|^2 \leq (1 + \eps)^2 \|q - p'\|^2$. Taking square roots reveals that $p$ is an $\eps$-ANN of $q$.
\end{proof}

Based on this lemma, we develop a new approximate Voronoi diagram (AVD) data structure for the \emph{well-separated case}. Consider a query region $w$, that is $\gamma$-separated from the point set $P$ for a suitably large constant $\gamma$ (to be specified below). We approximate the Voronoi subdivision within $w$ by building upon the Macbeath-based approach for approximating convex bodies~\cite{AbM18}. The benefits of this approach in our case are: (1) achieving state-of-the-art storage bounds, and (2) providing a covering of $w$ that facilitates blending.

\subsection{Distance-Based Macbeath Regions.} \label{sec:distance_based_macbeath}

In this section, we explain how to adapt the concept of a Macbeath region to approximate Voronoi diagrams. Macbeath regions are a classical structure that have been employed in the context of convex shape approximation and analysis (see, e.g., \cite{AbM18, AFM23, BaL88, Bar07}). Given a convex body $K$ and a point $x \in K$, the \emph{Macbeath region} at $x$ is defined to be
\[
    M_K(x)
        ~ = ~ K \cap (2 x - K)
        ~ = ~ x + ((K-x) \cap (x-K)).
\]
It is easy to verify that $M_K(x)$ is the intersection of $K$ and the reflection of $K$ around $x$, and hence it is the largest centrally-symmetric body contained in $K$ and centered at $x$. For $\lambda > 0$, the \emph{$\lambda$-scaled Macbeath region} at $x$ is defined to be a central scaling of $M_K(x)$ by a factor of $\lambda$ about $x$, that is,
\[
    M_K^{\lambda}(x) 
        ~ = ~ x + \lambda ((K-x) \cap (x-K))
\]
(see Figure~\ref{macbeath.fig}(a)).  $M_K^{\lambda}(x)$ is just a scaling of this shape by a factor of $\lambda$ about $x$. When $\lambda < 1$, we say $M_K^{\lambda}(x)$ is \emph{shrunken}.

\begin{figure}[htbp]
  \centerline{\includegraphics[scale=0.40,page=1]{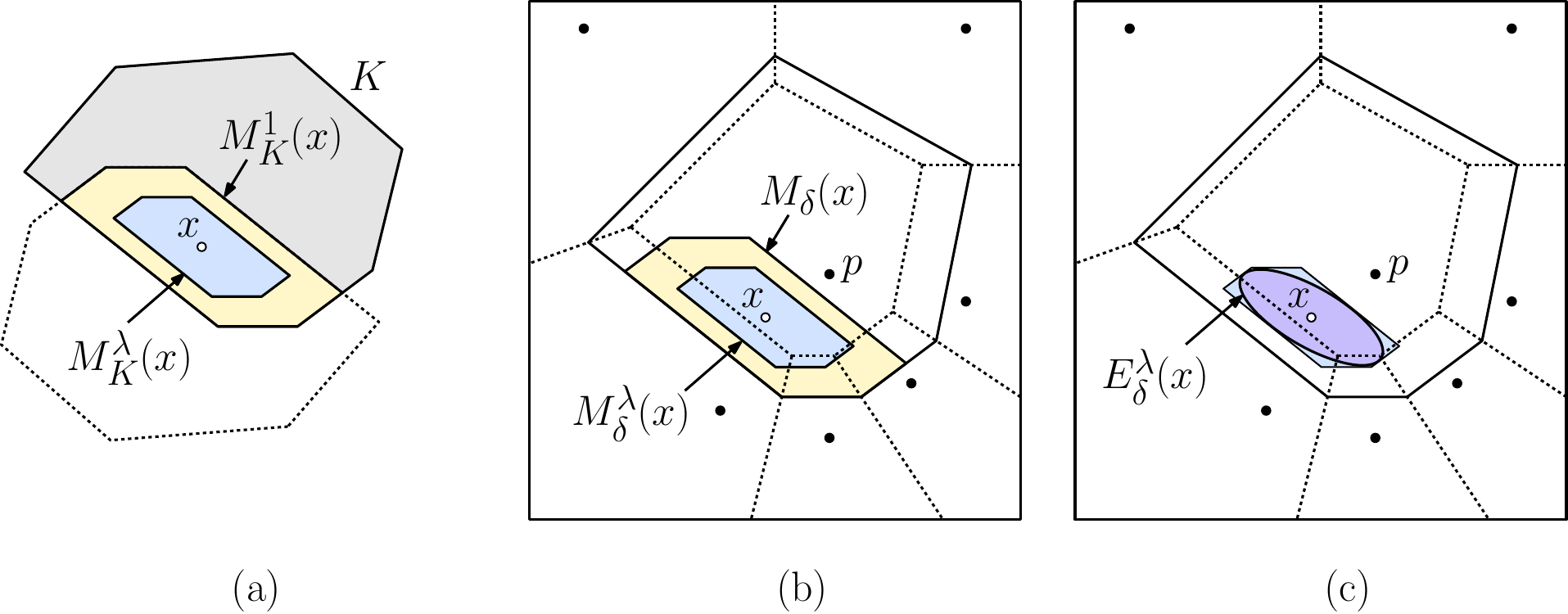}}
  \caption{(a) Macbeath region, (b) distance-based Macbeath region, and (c) Macbeath ellipsoid.} \label{macbeath.fig}
\end{figure}

Given $x \in \RE^d$ and a point set $P$, let $\nn_P(x)$ denote the nearest neighbor of $x$ in $P$, where ties are broken by index. When $P$ is clear from context, we will omit the subscript. Given $\delta > 0$, define the \emph{distance-based Macbeath region}, denoted $M_{\delta}(x)$, to be the Macbeath region of $x$ with respect to the expanded Voronoi cell $V_{\delta}(\nn(x))$ (see Figure~\ref{macbeath.fig}(b)). This is well-defined because $x \in V(\nn(x))$, which is a subset of $V_{\delta}(\nn(x))$. Given $\lambda \geq 0$, define $M_{\delta}^{\lambda}(x)$ to be a central scaling of $M_{\delta}(x)$ by a factor of $\lambda$. Similarly, define the \emph{distance-based Macbeath ellipsoid}, denoted $E_{\delta}(x)$, to be the maximum volume ellipsoid contained within $M_{\delta}(x)$, and define its scaling $E_{\delta}^{\lambda}(x)$ analogously (see Figure~\ref{macbeath.fig}(c)).%
\footnote{In spaces of constant dimension, the maximum-volume ellipsoid contained within a convex body is unique, and can be computed for a convex polytope $K$ in time linear in the number of its bounding halfspaces \cite{ChM96}.}
By John's Theorem~\cite{Bal97} we have
\begin{equation}
  E^{\lambda}_{\delta}(x)
	~ \subseteq ~ M^{\lambda}_{\delta}(x)
	~ \subseteq ~ E^{\lambda\sqrt{d}}_{\delta}(x). \label{eq:john-thm}
\end{equation}

\begin{figure}[htbp]
  \centerline{\includegraphics[scale=0.40,page=2]{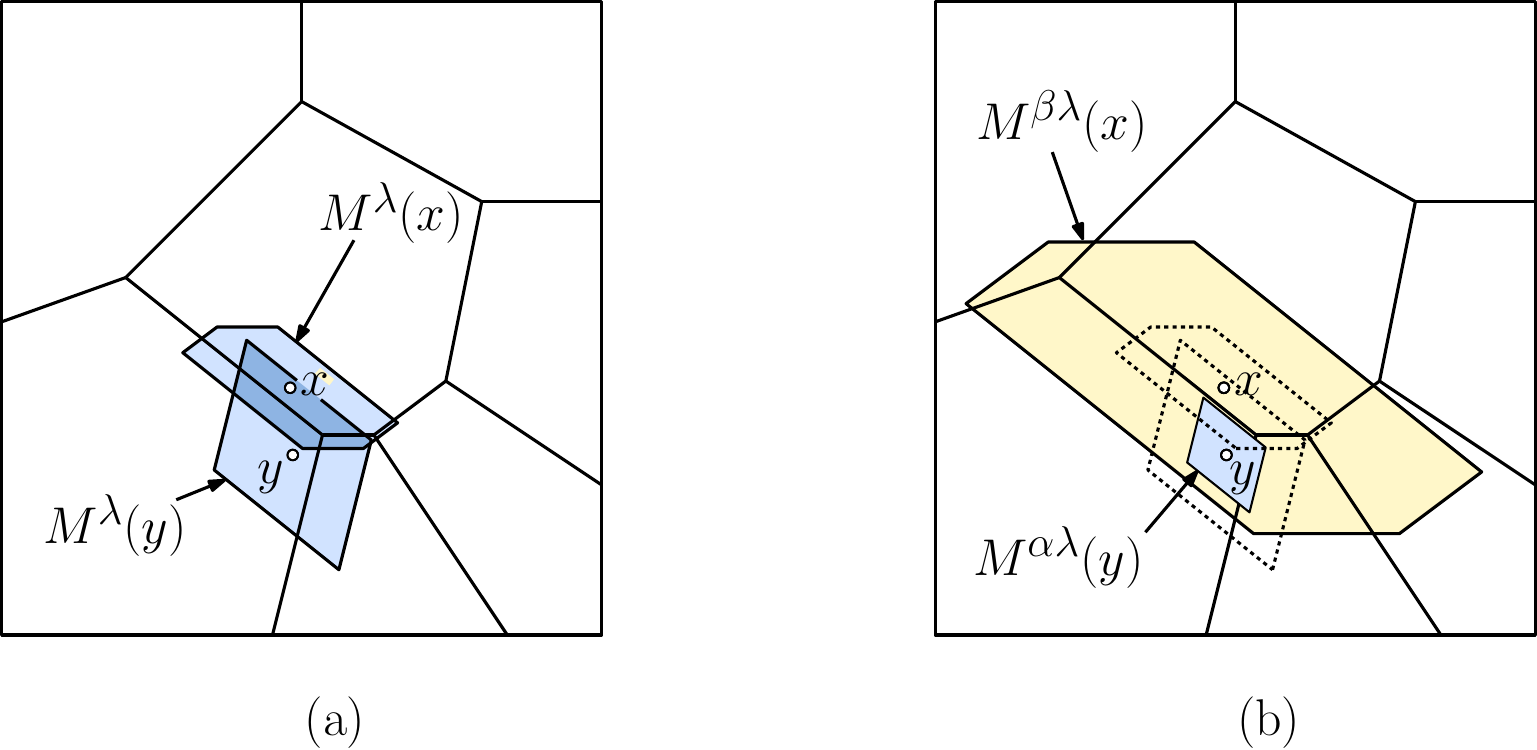}}
  \caption{(a) Overlapping Macbeath regions and (b) expansion-containment.} \label{macbeath-2.fig}
\end{figure}

To carry out the design of the \emph{ellipsoidal AVD} (EVD), we establish generalizations of two key properties of classical Macbeath regions for the new distance-based counterparts. The proofs, deferred to Sections~\ref{sec:vor-exp-con} and~\ref{sec:vor-ecc}, apply the well-known Voronoi lifting transformation (see, e.g., ~\cite{Boy98, BCK10}) to relate distance-based Macbeath regions in $\RE^d$ to the classic Macbeath regions of a lifted polytope in $\RE^{d+1}$.

First, we establish the \emph{expansion-containment} property, which ensures that whenever two shrunken Macbeath regions overlap, a constant-factor expansion of one contains the other (see Figure~\ref{macbeath-2.fig}).

\begin{restatable}{lemma}{VorExpCon} \label{vor-exp-con.lem}
Consider a finite set $P \subset \RE^d$ and two scalars $\delta$ and $\lambda$, where $\delta \geq 0$ and $0 < \lambda < 1$. If $x,y \in \RE^d$ such that $M_{\delta}^{\lambda}(x) \cap M_{\delta}^{\lambda}(y) \neq \emptyset$, then for any $\alpha \geq 0$ and $\beta \geq 2\cdot\frac{2 + \alpha(1+\lambda)}{1 - \lambda}$, $M^{\alpha\lambda}(y) \subseteq M^{\beta\lambda}(x)$.
\end{restatable}

Second, we establish a bound on the complexity of the subdivision. The proof follows from existing bounds on the size of Macbeath-based covers (see, e.g., \cite{AFM17c}), together with the close relationship between distance-based and classical Macbeath regions. Remarkably, the bound depends only on $\eps$ and the separation parameter $\gamma$ and is independent of the number of sites.

\begin{restatable}{lemma}{rsVorECC}\label{vor-ecc.lem}
Consider a point set $P$ and a query region $w$ in $\RE^d$ which are concentrically $2$-separated. Let $b$ be the ball containing the inner set. For any positive constant $\lambda$ and positive scalar $\gamma = O(1/\eps)$, let $X$ be a maximal set of points lying within $w \cap \gamma b$ such that the ellipsoids $E^{\lambda}_{\delta}(x)$ are pairwise disjoint, where $\delta = \gamma \sqrt{\eps} \cdot \radius(b)$. Then $|X| = O\big(\big(1/(\gamma\eps))^{d/2}\big)$.
\end{restatable}

\subsection{The EVD Data Structure.} \label{sec:evd}

Consider a query region $w \subseteq \RE^d$, and a set of points $P$, such that $w$ and $P$ are concentrically 2-separated (recall Figure~\ref{separation_lite.fig}). Let $b$ denote the Euclidean ball containing the inner set so that the other set lies outside $2 b$. Our structure will involve collections of ellipsoids arranged into levels, where each level covers a certain region of interest. For $i \geq 0$, define $\gamma_i = 2^i$, $b_i = \gamma_i b$, and $\delta_i = \gamma_i \sqrt{\eps} \cdot \radius(b)$. Define $w_i = w \cap b_{i+1}$ to be the portion of the query region handled by level $i$ of the structure. Observe that when $w$ is the inner set, $w_i = w$, for all $i$ (see Figure~\ref{evd-setup.fig}(a)). On the other hand, if $w$ is the outer set, the sets $w_i$ are the intersection of $w$ with concentric Euclidean balls of exponentially increasing radii (see Figure~\ref{evd-setup.fig}(b)). While there is no \textit{a priori} bound on how large $w$ can be, we may restrict attention to the portion of $w$ that lies within the ball $(3/\eps) b$.

\begin{figure}[htbp]
  \centerline{\includegraphics[scale=0.40]{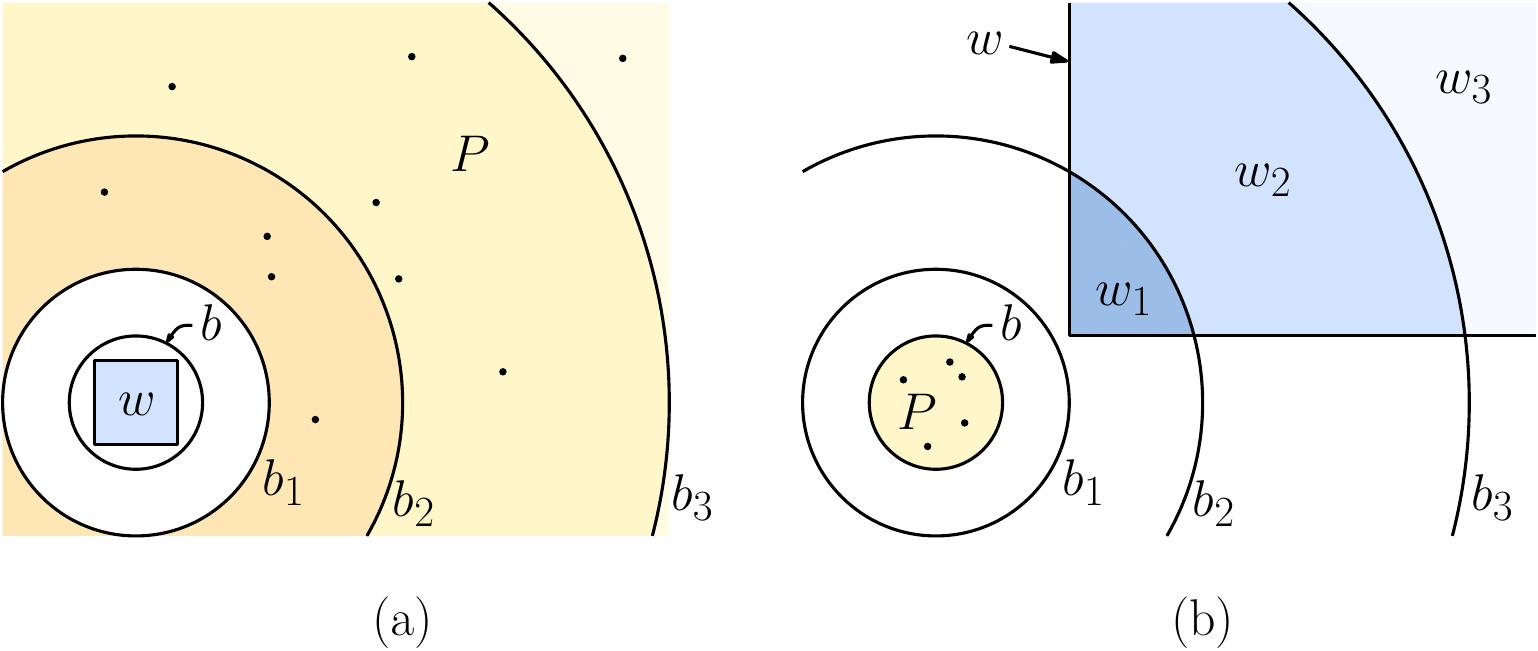}}
  \caption{EVD structure when $w$ and $P$ are 2-separated: (a) $w$ is the inner set and (b) $P$ is the inner set.} \label{evd-setup.fig}
\end{figure}

For any point $x \in w_i$, consider the ellipsoids $E'_i(x) = E^{1/2}_{\delta_i}(x)$ and $E''_i(x) = E_{\delta_i}^{\lambda_1}(x)$, where $\lambda_1 = 1/(16 \sqrt{d} + 1)$. Let $X_i$ be any maximal set of points, all lying within $w_i$, such that the ellipsoids $E''_i(x)$ are pairwise disjoint. Clearly, the smaller ellipsoids $E''_i(x)$ form a packing of $w_i$, and it follows (essentially from the expansion-containment property of Lemma~\ref{vor-exp-con.lem}) that the larger ellipsoids $E'_i(x)$ form a cover of $w_i$ (see, e.g., ~\cite{AbM18}). See Figure~\ref{ellipsoid-cover.fig}.

\begin{figure}[htbp]
  \centerline{\includegraphics[scale=0.40]{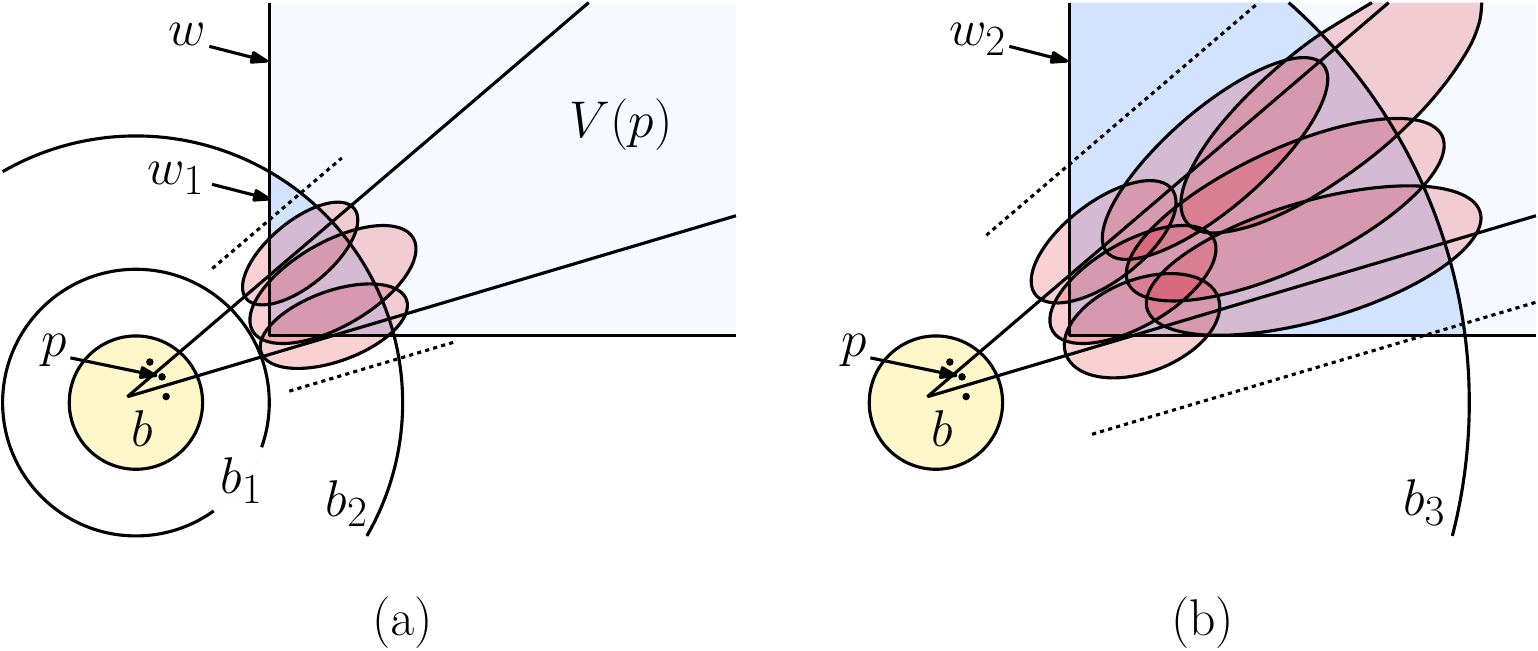}}
  \caption{(a) The ellipsoid cover for $w_2 \cap V(p)$ and (b) the ellipsoid cover for $w_3 \cap V(p)$.} \label{ellipsoid-cover.fig}
\end{figure}

There are two stopping criteria. If $w$ is the inner set, we stop at any level $\ell$ such that $|X_{\ell}| = 1$ (a single ellipsoid covers the entire query region). If $w$ is the outer set, we stop at the first level $\ell$ such that $|X_{\ell}| = 1$ and $(w \cap (3/\eps) b) \subseteq b_{\ell}$ (a single ball covers the relevant portion of the query region). We show below that such an $\ell$ exists in either case.

The data structure is a rooted directed acyclic graph (DAG), where the levels are indexed from $\ell$ (the root) down to 1 (the leaves). The nodes of level $i$ of the DAG correspond 1--1 to the points of $X_i$. We also identify $X_i$ with level $i$. Each $x \in X_i$ is associated with: 
\begin{itemize} \setlength{\itemsep}{3pt}\setlength{\parsep}{2pt}%
\item a \emph{cell}: $\cell(x) = E'_i(x)$,

\item a \emph{representative site}: $\rep(x) = \nn(x)$,

\item \emph{children}: for $i \geq 2$, $\ch(x)$ consists of the nodes $y \in X_{i-1}$ such that $\cell(x) \cap \cell(y) \neq \emptyset$.
\end{itemize}

A nearest-neighbor query for a point $q \in w$ is answered by a simple descent of the EVD. By the remarks made above, we may assume that $q$ lies within $w \cap (3/\eps)b$, since otherwise we may report any point of $P$ as an $\eps$-approximate nearest neighbor. Therefore, $q$ lies within the root's cell. To answer the query, we descend the DAG level by level, visiting any child of the current node whose cell contains the query point until no such child exists (or the leaf level is reached). Let $x$ be the last node visited by the search. Then, $\rep(x)$ is a valid $\eps$-approximate nearest neighbor of $q$.

Intuitively, each successive level of the DAG from the root down involves an ellipsoid cover based on exponentially smaller values of $\delta_i$, implying that the representatives of the ellipsoids containing the query point at successive levels provide successively better approximations to the nearest neighbor. The following lemma shows that when the search procedure terminates, the representative of the last node visited will be an $\eps$-ANN of the query point.

\begin{restatable}{lemma}{rsCorrectness}\label{ann-correctness.lem}
Given a point set $P$ and query region $w$ that are concentrically 2-separated and $0 < \eps \leq 1/2$, the above search algorithm returns an $\eps$-ANN among the points of $P$ for any query point $q \in w$.
\end{restatable}

\begin{proof}
Consider any query point $q \in w$. As observed earlier, if $w$ is the outer set and $q$ lies outside $(3/\eps) b$, then we may report any point of $P$ as $q$'s $\eps$-ANN. Otherwise, $q$ lies within the cell of the root of the $\eps$-EVD. The search procedure descends the DAG until encountering a node that has no child whose cell contains $q$. Let $i \geq 1$ denote this node's level, let $x \in X_i$ be the point associated with this node, and let $p = \rep(x) = \nn(x)$. Because $q \in \cell(x)$ and $\cell(x) = E^{1/2}_{\delta_i}(x) \subseteq V_{\delta_i}(p)$, we have $q \in V_{\delta}(p)$. 

The nodes of $X_{i-1}$ cover $w_{i-1} = w \cap b_i$, and thus $q$ lies outside $b_i$. (When $i = 1$, this is trivially true by concentric 2-separation.) Let $r = \radius(b)$. By the triangle inequality, $\|p - q\| \geq (\gamma_i - \gamma_0) r \geq \gamma_i r/2$. By definition, $\delta_i = \gamma_i \sqrt{\eps} r/2$, and hence 
\[
	\frac{\delta_i}{\|p - q\|} 
		~ \leq ~ \frac{\gamma_i \sqrt{\eps} r/2}{\gamma_i r/2}
		~  =  ~ \sqrt{\eps}. 
\]
By Lemma~\ref{exp-vor.lem} and the fact that $q \in V_{\delta}(p)$, $p$ is an $\eps$-ANN of $q$, as desired.
\end{proof}

The query time depends on the product of the number of levels $\ell$ and the maximum out-degree of each node. In the following lemma, we show that the number of levels is $O(\log (1/\eps))$. Intuitively, this holds due to the doubling of radii at each level, combined with the fact that once the distance exceeds $O(\radius(b)/\eps)$, a single representative suffices. The technical details are deferred to Section~\ref{sec:evd-height}.

\begin{restatable}{lemma}{rsRootLevel}\label{root-level.lem}
Given a point set $P$ and query region $w$ that are concentrically 2-separated and $0 < \eps \leq 1/2$, the $\eps$-EVD structure has $O(\log (1/\eps))$ levels. 
\end{restatable}

The following lemma bounds the number of children. This follows in a straightforward manner from the expansion-containment properties of the distance-based Macbeath regions and a simple packing argument. The proof is presented in Section~\ref{sec:ann-child-bound}.

\begin{restatable}{lemma}{rsChildBound}\label{ann-child-bound.lem}
Each node of the $\eps$-EVD has at most $\big(14\sqrt{d}\big)^d = O(1)$ children.
\end{restatable}

Finally, we bound the total space used by the data structure. This follows by applying Lemma~\ref{vor-ecc.lem} to each of the successive levels. The sequence forms a geometric series, and hence the total complexity is dominated by the complexity of the leaf level. The proof is also presented in Section~\ref{sec:ann-child-bound}.

\begin{restatable}{lemma}{rsAnnSpaceBound}\label{ann-space-bound.lem}
Given a point set $P$ and query region $w$ that are concentrically 2-separated and $0 < \eps \leq 1/2$, the total storage required by the $\eps$-EVD is $O(1/\eps^{d/2})$.
\end{restatable}

Even without smoothing, we have shown that, in the well-separated case, $\eps$-ANN queries can be efficiently processed by a simple descent in a DAG structure based on a hierarchy of ellipsoids.

\begin{theorem} \label{thm:evd_ws}
Given a query region $w \subseteq \RE^d$ and a set of points $P$, such that $w$ and $P$ are concentrically 2-separated, the EVD data structure can answer $\eps$-approximate nearest neighbor queries, for an approximation parameter $0 < \eps \leq 1/2$ and any $q \in w$ against $P$ with
\[
    \hbox{Query time:~} O\kern-2pt \left(\log 1/\eps\right)
	\quad\hbox{and}\quad
    \textrm{Space:~} O\kern-2pt \left(1 / \eps^{d/2} \right).
\]
\end{theorem}

\section{General ANN: Quadtree-EVD} \label{sec:quadtree_evd}

In this section, we integrate the separation-specialized EVD data structure from Section~\ref{sec:evd} into a well-known AVD data structure to obtain a new data structure for general $\eps$-ANN queries. Our data structure will match state-of-the-art bounds on its query time and storage. The resulting data structure, which we call the \emph{Quadtree-EVD}, can be applied in a context where continuity is not needed. In Section~\ref{sec:blending} we will show how to incorporate blending into this structure in order to yield the desired continuity guarantees on the encoded distance field.

Before presenting our data structure, we begin by recalling the AVD data structure presented in \cite{AMM09a}. It employs a height-balanced variant of a quadtree, called a \emph{balanced box decomposition tree}, or BBD tree~\cite{ArM00}. A BBD tree is a hierarchical structure for storing a point set $P$. Each node of the tree is associated with a geometric region, called its \emph{cell}. Each cell is either a quadtree box or the set-theoretic difference of two quadtree cells, an \emph{outer box} and an optional \emph{inner box}. 

The construction subdivides space so that the leaf cells satisfy basic separation properties with respect to the point set $P$. For each cell $w$, let $b_w$ denote a tight enclosing ball of radius $\diam(w)/2$ whose center coincides with the center of $w$'s outer box. For any leaf cell $w$ of the AVD, the sites can be partitioned into three subsets, any of which may be empty. First, a single site may lie within $w$. Second, a subset of sites, called the \emph{outer cluster}, is well-separated from the cell. Finally, there may be a dense cluster of points, called the \emph{inner cluster}, that lie within a ball $b'_w$ that is well-separated from $w$. The next lemma encapsulates the essential AVD separation properties. It follows from the construction given in \cite[Lemma 6.1]{AMM09a}. For completeness, a proof is presented in Section~\ref{sec:avd-blending}.

\begin{figure}[htbp]
  \centerline{\includegraphics[scale=0.40]{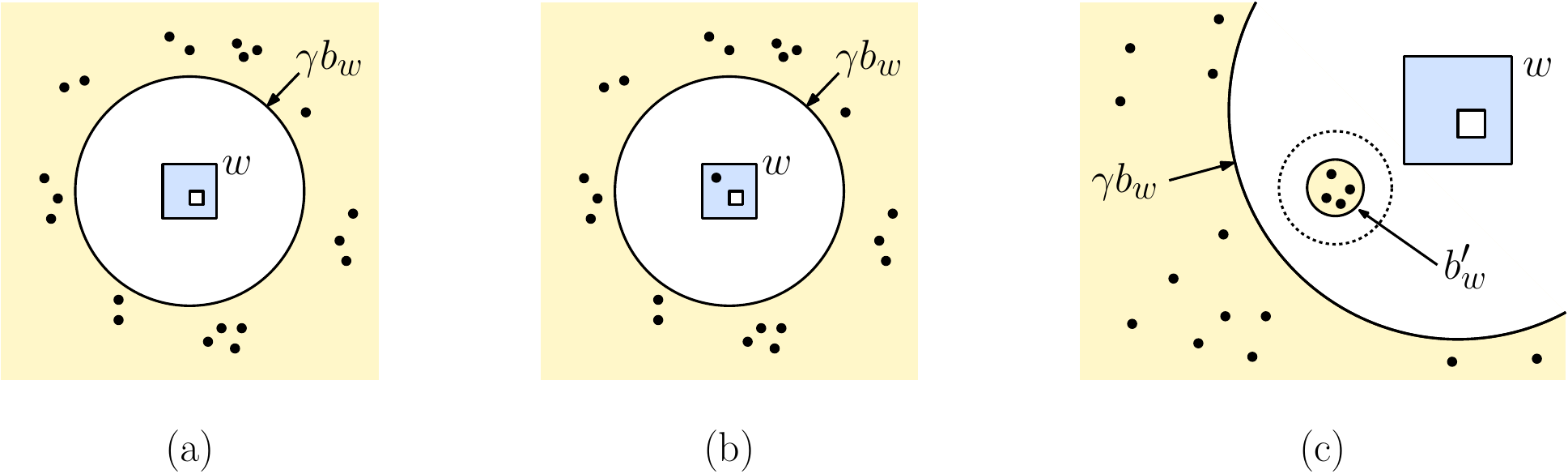}}
  \caption{AVD Separation properties.}
  \label{separation.fig}
\end{figure}

\begin{restatable}{lemma}{separationProps}\emph{[AVD separation properties]}\label{sep-props.lem}
Given a set $P$ of $n$ points in $\RE^d$ and a constant $\gamma \geq 5$, there exists a BBD tree with $O(n)$ nodes storing the points of $P$, where each leaf cell $w$ satisfies at least one of the following separation properties:
\begin{enumerate}
\item[$(a)$] $P \cap \gamma b_w = \emptyset$, implying that $w$ is concentrically $\gamma$-separated from $P$ (see Figure~\ref{separation.fig}(a)).

\item[$(b)$] $\gamma b_w$ contains exactly one point of $P$, which lies within $w$. This point is the nearest neighbor of any query point in $w$ (see Figure~\ref{separation.fig}(b)).

\item[$(c)$] There exists a ball $b'_w$ contained within $\gamma b_w$ that is concentrically $2$-separated from $w$, and the nearest neighbor of any point within $w$ lies within $P \cap b'_w$ (see Figure~\ref{separation.fig}(c)).
\end{enumerate}
Furthermore, there exists a constant $c$ such that for any $x \in w$, $\radius(b_w) \geq c \cdot d_P(x)$. The tree can be constructed in $O(n \log n)$ time, and the cell containing a given query point can be found in $O(\log n)$ time.
\end{restatable}

The top level of our Quadtree-EVD is the AVD structure described in this lemma for the point set $P$ and $\gamma = 2$. Each leaf cell $w$ of this structure is associated with an additional data structure that answers $\eps$-ANN queries for any point within the cell. Recall that we assume $0 < \eps \leq 1/2$. If property (b) holds, all the points of $w$ share the same Euclidean nearest neighbor, and the associated data structure is trivial. In case (a), the cell $w$ is concentrically $2$-separated from its Euclidean possible nearest neighbors in $P$, and hence we attach an $\eps$-EVD structure for $w$ and $P$. Finally, in case (c) the possible Euclidean nearest neighbors of $w$ lie within the ball $b'_w$ that is concentrically $2$-separated from $w$. Again, we attach an $\eps$-EVD structure for $w$ and $P \cap b'_w$. 

Combining the $O(n)$ space of the AVD with the $O(1/\eps^{d/2})$ space of each EVD, we have a total space bound of $O(n/\eps^{d/2})$. Also, summing the $O(\log n)$ query time for the AVD with the $O(\log 1/\eps)$ query time for the EVD, we obtain an overall query time of $O(\log (n/\eps))$. In summary, we have:

\begin{theorem} \label{thm:evd}
Given a set $P$ of $n$ points in $\RE^d$, an approximation parameter $0 < \eps \leq 1/2$, the Quadtree-EVD data structure can answer $\eps$-approximate nearest neighbor queries with
\[
    \hbox{Query time:~} O\kern-2pt \left(\log n/\eps\right)
	\quad\hbox{and}\quad
    \textrm{Space:~} O\kern-2pt \left(n / \eps^{d/2} \right).
\]
\end{theorem}

It is noteworthy that this matches the space and query times given by Arya {\etal}~\cite{AFM17a} for $\eps$-ANN queries, but is simpler since it involves a simple descent through a hierarchical structure and avoids the dependence on a separate ray-shooting data structure.

\section{The Smoothed Quadtree-EVD} \label{sec:blending}

The plain Quadtree-EVD presented in Section~\ref{sec:quadtree_evd} is a valid $\eps$-ANN data structure assuming a witness-based approach. However, as mentioned in the introduction, this cannot be used as the basis of a smooth distance approximation. In this section, we show how to enhance the plain Quadtree-EVD data structure to enable blending as described in Section~\ref{sec:pou}, where we seek to eliminate all discontinuities while only incorporating valid distance approximations no greater than a $(1 + \eps)$ factor of the true distance function.

In order to achieve this, each leaf-level node stores a collection of overlapping ellipsoids from the $\eps$-EVD data structures attached to the neighboring quadtree cells. Lemma~\ref{avd-lookup-bound.lem} below shows that this replication increases the space bounds by only a constant factor. The resulting collection of cells and ellipsoids for a given query point provides the patches over which we perform the blending computations in Eq.~\eqref{eq:pu} (recall Figure~\ref{AVD_blending.fig}.)

\subsection{Smoothing the EVD for Blending.} \label{sec:evd_blending}

As described in Section~\ref{sec:evd}, each Macbeath ellipsoid $E'_i(x)$, with the center point $x \in X_i$, is defined within some level $i$ with respect to some leaf quadtree cell $w$. In addition, when $w$ is an outer cell, the ellipsoid $E'_i(x)$ may only be at the leaf level for point $q \in 2^{i+1} b \setminus 2^{i}b$, where $b$ is the Euclidean ball containing the inner set with $w$ lying outside $2b$.

For an outer quadtree cell $w$ with an inner set contained in the ball $b$, and an EVD level $i \geq 1$, we avoid discontinuities at the boundary of $2^i b$, by restricting ellipsoids $E'_i(x)$ within the ring $R_i = (5\cdot2^{i-1})b\setminus(3\cdot2^{i-2})b$. Because $3\cdot2^{i-2} < 2^i < 2^{i+1} < 5\cdot2^{i-1}$, it follows that $R_i$ contains $2^{i+1} b \setminus 2^{i}b$. This allows the blending to extend well into the adjacent rings, but not beyond. To further improve the bound on $c_\Psi$ in Lemma~\ref{positive_Psi.lem} below, we use $\lambda_1 = 1/(32\sqrt{d}+1)$, which is smaller than the value used in Section~\ref{sec:evd}.

We employ the ellipsoid and ring shape functions $\shape^{[e]}$ and $\shape^{[r]}$ from Eqs.~\eqref{eq:shape-e} and~\eqref{eq:shape-r}. For each Macbeath ellipsoid, the associated shape function $\shape^{[e]}$ is defined using the matrix $M$ defining the ellipsoid. For a ring $R_i$, defined with respect to the ball $b$ containing an inner set with radius $r$, we set $a = 5\cdot 2^{i-1}r$ and $b = 3\cdot 2^{i-2}r$. The following lemmas show that these shape functions are compatible. The proofs appear in Section~\ref{sec:compatibility}.

\begin{restatable}{lemma}{ellipsoidCompatibility}\label{ellipsoid-compatibility.lem}
The shape functions $\sigma^{[e]}$ associated with all Macbeath ellipsoids are compatible.
\end{restatable}

\begin{restatable}{lemma}{ringCompatibility}\label{ring-compatibility.lem}
The shape functions $\sigma^{[r]}$ associated with the rings of the EVD are compatible.
\end{restatable}

\subsection{Smoothing the BBD for Blending.} \label{sec:quadtree_blending}

For a given query point $q$, let $W(q)$ denote the set of leaf quadtree cells used for blending, that is,
\[
    W(q) 
        ~ = ~ \{w \ST q \in 2\cdot b_w\}.
\]

\begin{restatable}{lemma}{avdLookup}\label{avd-lookup-bound.lem}
Given any query point $q \in \RE^d$, $|W(q)| = O(1)$.
\end{restatable}

For any quadtree cell $w \in W$, we employ the ball shape function $\shape^{[b]}$ from Eq.~\eqref{eq:shape-b}, with the center point $c^{[b]}$ set to the center of $w$ and the radius $r$ being twice the radius of $b_w$. Regardless of how the query for $q$ is answered by $w$, we sum the result into blending multiplied by the factor $\mu(\shape^{[b]}(q))$ before normalizing (see Eqs.~\eqref{sDF.eq} and~\eqref{eq:psi_i}).

\begin{restatable}{lemma}{quadtreeCompatibility}\label{quadtree-compatibility.lem}
The shape functions $\sigma^{[b]}$ associated with leaf-level quadtree cells are compatible.
\end{restatable}

\subsection{Putting It Together.} \label{sec:putting_it_together}
The smoothed query data structures above have the following properties. Combining Lemma~\ref{avd-lookup-bound.lem} with Lemma~\ref{ann-child-bound.lem}, we establish that we need only blend a constant number of local approximations for any given query point.

\begin{lemma} \label{PU_depth.lem}
The smoothed Quadtree-EVD construction satisfies $\depth = O(1)$.
\end{lemma}

Second, we ensure a positive lower-bound $c_\Psi$ on $\Psi$, the sum of blending weights used for normalization. As a consequence of (b), any $z \in \RE^d$ is not only guaranteed to be covered by some ellipsoid $E'_{\delta_k}(x) = E^{1/2}_{\delta_k}(x)$, but that $z \in E^{1/4}_{\delta_k}(x)$. It follows that the function $f^{[e]}$ associated with that ellipsoid $E'_{\delta_k}(x)$ satisfies $f^{[e]}(z) \leq 1/2$. Similarly, as $z$ must belong to some quadtree cell $w$, which is fully contained in its enclosing ball $b_w$, the function $f^{[b]}$ associated with the $2b_w$ satisfies $f^{[b]}(z) \leq 1/4$. Finally, when $z$ belongs to a ring $5\cdot 2^{j+1}b'_w \setminus 3\cdot 2^{j-1}b'_w$ for some cell AVD $w$, with the radii specified in (d), we have $|f^{[r]}(x)| \leq 3/4$ for one of the rings containing $z$ (see Section~\ref{sec:compatibility}). Letting $\patch_i$ denote the patch realized by the intersection of the ellipsoid denoted above by $E'_{\delta_k}(x)$ restricted to the ball $b_w$, and possibly the ring $5\cdot 2^{j+1}b'_w \setminus 3\cdot 2^{j-1}b'_w$, we have $\Psi(z) \geq \psi_i(x) = \mu(\frac{1}{2})\cdot\mu(\frac{1}{4})\cdot\mu(\frac{3}{4}) > 0.009$. Hence, we obtain the following result.

\begin{lemma} \label{positive_Psi.lem}
The smoothed Quadtree-EVD construction ensures that $c_\Psi > 0.009$.
\end{lemma}

Finally, the ANN problem is routinely solved for a subset of space, e.g., a sufficiently large ball containing all data points, where any point becomes a valid approximate nearest neighbor outside that ball. To ensure continuity for all $\RE^d$, we assign the same representative to all AVD cells at the periphery of this ball. In this manner, only a single distance function is effectively considered for blending at the periphery, and we can continue to evaluate this single function away from the ball without the need for blending while ensuring continuity.

\section{Conclusions} \label{sec:conclusion}
In this paper, we presented an efficient differentiable approximation to the distance field $\DF$ induced by a set of points $P \subset \RE^d$. Given a query point $q$, the proposed data structure computes an approximate value $\sDF(q)$ in addition to its gradient $\Gradient \sDF(q)$. This additional feature enables distance approximation to be integrated into gradient-based optimization pipelines, e.g., for machine learning. The resulting data structure matches state-of-the-art results for approximate nearest-neighbor searching in terms of both the query time and storage complexity, while additionally achieving asymptotically optimal bounds on the norms of both the gradient and the Hessian of the approximate distance function. Our approach incorporates a partition of unity construction into an approximate Voronoi diagram data structure based on a new type of Macbeath regions defined with respect to the convex Voronoi subdivisions.

Beyond distance approximation, the partition of unity framework, with the flexibility to blend over various shape functions, is general and can be applied to other problems in geometric approximation. Examples include penetration depth in collision detection~\cite{ZKM14}, distance oracles in robotics and autonomous navigation \cite{WWL22}, and novel-view synthesis using parametric radiance fields~\cite{FYT22}.

While our approach produces a smooth differentiable approximation, there are other properties of distance fields that would be useful to preserve. One shortcoming of our method is that it can produce spurious local minima in the approximate distance field. An interesting question is whether our approach can be modified to eliminate these minima. In addition, while we bound the value of the approximate distance function $\sDF$ by a multiplicative $(1 + \eps)$ factor of the exact distance function $\DF$, we only bound the norm of the gradient $\|\Gradient \sDF(x)\|$ as $O(1)$, which can be much larger than $\|\Gradient \DF(x)\| = 1$. It would be interesting to investigate improved approximations to the gradient and their implications on the storage and query time.

\section*{Acknowledgments} \label{sec:ack}
The first author would like to thank Jens M. Melenk for a fruitful discussion of this research during his visit at the University of Texas, Austin, and in particular for providing the reference to Stein's classical result~\cite{Ste70}.

\appendix

\section{Deferred Proofs} \label{sec:def-proof}

\subsection{Partition of Unity Derivations.} \label{sec:pu-math}

Recall the we utilize a locally-finite open cover of $\RE^d$ by a collection of patches $\{\patch_i\}_{i \in I}$ with a partition of unity $\{\phi_i\}_{i \in I}$ subordinate to the cover, that is, $\supp(\phi_i) \subseteq \patch_i$. Recall that each patch is associated with its nearest-neighbor representative point from $P$, denoted $\rep(\patch_i)$, and a local approximation function, $\localF_i(x) = \|x - \rep(\patch_i)\|^2$. The approximate distance function is defined as:
\[
    \sDFSq(x) 
        ~ = ~ \sum_{i \in I} \phi_i(x)\cdot \localF_i(x).
\]

In this section, we derive various expressions and bounds on the gradients of the partition of unity and its constituent functions. This culminates in a simple expression of the gradient of the distance approximation $\Gradient \sDFSq(x)$ along with an upper bound on its magnitude. Let us begin by recalling the various elements of the partition of unit approach.

\subsubsection{Functions.}

Recall the bump function $\mu: \RE \rightarrow \RE_{\geq 0}$ from Eq.~\eqref{eq:bump}, with $\supp(\mu) \subseteq [0, 1]$.
\[
    \mu(s) 
        ~ = ~ \begin{cases}
                \exp\left(\dfrac{1}{s^2 - 1}\right),   & |s| < 1,\\
                0,                                     & \text{otherwise.}
            \end{cases}
\]
In what follows, the support of all functions is similarly bounded, and we omit explicit references to the complement of the support. We obtained the weight functions defining the partition of unity over a set of patches $\{\patch_i\}$, where each patch is realized as the intersection of a set of at most $\kappa_\Pi = O(1)$ \emph{component shapes} $\{\mathcal{S}_i^{(k)}\}$ and each component shape is associated with a \emph{compatible shape function} $\sigma_i^{(k)}$ (recall Section~\ref{sec:shape}).

A partition of unity is the set of functions $\bigg\{\phi_i = \dfrac{\psi_i}{\Psi}\bigg\}_{i \in I}$, such that:
\begin{equation} \label{eq:psi_Psi}
    \psi_i(x) 
        ~ = ~ \prod_{k \in k_i} \mu\left(\shape_i^{(k)}(x)\right) \text{~~ and ~~} 
    \Psi(x) 
        ~ = ~ \sum_{i \in I} \psi_i(x).
\end{equation}
%

\subsubsection{Blending Gradients.}

For a given function $s: \RE^d \rightarrow \RE$, it is easily verified that
\[ 
    \Gradient \mu(s) 
        ~ = ~ \exp\left(\frac{1}{s^2 - 1}\right)\cdot\Gradient\left(1+\frac{1}{s^2-1}\right) 
        ~ = ~ \mu(s)\cdot\frac{-2s}{(s^2-1)^2}\cdot\Gradient s.
\]
From the above equations, we have
\begin{align*}
    \Gradient \phi_i(x) 
        & ~ = ~ \frac{1}{\Psi(x)^2}\big(\Gradient \psi_i(x) \cdot \Psi(x) - \psi_i(x)\cdot\Gradient\Psi(x)\big) 
          ~ = ~ \frac{\Gradient \psi_i(x)}{\Psi(x)} - \frac{\psi_i(x)}{\Psi(x)^2}\cdot\sum_j\Gradient\psi_j(x), \\
        & ~ = ~ \frac{1}{\Psi(x)} \bigg(\Gradient \psi_i(x) - \phi_i(x)\sum_j\Gradient\psi_j(x)\bigg), \text{~and}\\
    \Gradient \localF_i(x) 
        & ~ = ~ 2\cdot(x - \rep(\patch_i)).
\end{align*}
Eq.~\eqref{eq:psi_Psi} is the product of exponentials. Differentiating and simplifying yields
\begin{equation} \label{eq:grad_wi}
    \Gradient \psi_i(x) 
        ~ = ~ \psi_i(x) \cdot \sum_{k \in k_i} \dfrac{-2\cdot\shape_i^{(k)}(x)}{\left(\shape_i^{(k)}(x)^2 - 1 \right)^2} \cdot \Gradient\shape_i^{(k)}(x).
\end{equation}

Towards the desired bound on the gradient of the distance approximation $\|\Gradient \sDFSq(x)\|$, we start by bounding the gradient of the weight functions $\|\Gradient \psi_i\|$.

\begin{restatable}{lemma}{blendingGradientBound}\label{blending_gradient_bound.lem}
For any set of compatible shape functions, we have $\|\Gradient \psi_i(x)\| = O(1/(\eps\cdot\DF(x)))$. The same bound holds for $\|\Gradient \phi_i(x)\|$.
\end{restatable}

\begin{proof}
We begin by simplifying Eq.~\eqref{eq:grad_wi}. We observe that for any $k \in k_i$,
\begin{align}
    \left| \frac{\psi_i(x)\cdot \shape_i^{(k)}(x)}{(\shape_i^{(k)}(x)^2-1)^2} \right| 
        ~ \leq ~ \mu(0)\cdot\left| \frac{\mu(\shape_i^{(k)}(x))\cdot \shape_i^{(k)}(x)}{(\shape_i^{(k)}(x)^2-1)^2} \right| 
        ~ \leq ~ \frac{1}{e}\cdot\frac{1}{2} 
        ~ <    ~ \frac{1}{2}, \label{eq:grad_bound}
\end{align}
which holds by construction as $|\shape_i^{(k)}(x)| < 1$ whenever $x \in \patch_i$.

By the triangle inequality and the properties of compatible shape functions, we have
\begin{equation} \label{nabla_psi_2.eq}
    \|\Gradient \psi_i(x)\| 
        ~ \leq ~ \sum_{k \in k_i} \left\|\Gradient \shape_i^{(k)}(x) \right\| 
        ~ \leq ~ \kappa_\Pi \cdot \max_k \left\|\Gradient \shape_i^{(k)}(x) \right\| 
        ~ =    ~ O\left(\dfrac{1}{\eps\cdot\DF(x)}\right),
\end{equation}
which completes the proof.
\end{proof}

\subsubsection{Smooth Distance Gradients.} \label{sec:sDFGradient}
In this section, we establish an explicit expression on the gradient of the smooth distance approximation, which can be easily evaluated in the course of answering distance-function queries, along with an upper bound on its gradient establishing the following lemma.

\sDFGradient*

\begin{proof}
Recalling that $\sDF(x) = \sum_{i \in I} \phi_i(x)\localF_i(x)$, where $\{\phi_i\}$ is the partition of unity and $\{\localF_i\}$ are the local approximations, we proceed to derive the gradient of $\sDF(x)$ as follows.
\begin{align*}
    \Gradient \sDFSq(x) 
        & ~ = ~ \Gradient \left(\sum_{i \in I} \phi_i(x)\cdot \localF_i(x)\right)
          ~ = ~ \sum_i \Gradient\phi_i(x)\cdot \localF_i(x) + \sum_i \phi_i(x)\cdot\Gradient \localF_i(x) \\
        & ~ = ~ \sum_i \left(\frac{\Gradient \psi_i(x)}{\Psi(x)} - \sum_{j \in I}\frac{\psi_i(x)}{\Psi(x)^2}\cdot\Gradient\psi_j(x) \right)\cdot \localF_i(x) + \sum_i \phi_i(x)\cdot\Gradient \localF_i(x) \\
        & ~ = ~ \sum_i \frac{\localF_i(x)}{\Psi(x)}\cdot\Gradient \psi_i(x) - \sum_j \left(\sum_i \frac{\psi_i(x)}{\Psi(x)}\cdot \localF_i(x)\right)\cdot\frac{\Gradient\psi_j(x)}{\Psi(x)} + \sum_i \phi_i(x)\cdot\Gradient \localF_i(x) \\
        & ~ = ~ \sum_i \frac{\localF_i(x)}{\Psi(x)}\cdot\Gradient \psi_i(x) - \sum_i \frac{\sDFSq(x)}{\Psi(x)}\cdot\Gradient\psi_i(x) + \sum_i \phi_i(x)\cdot\Gradient \localF_i(x) \\
        & ~ = ~ \sum_i \frac{\localF_i(x)-\sDFSq(x)}{\Psi(x)}\cdot\Gradient\psi_i(x) + \sum_i \phi_i(x)\cdot\Gradient \localF_i(x) \\
        & ~ = ~ \sum_i \frac{\localF_i(x)-\sDFSq(x)}{\Psi(x)}\cdot\Gradient\psi_i(x) + \sum_i 2\cdot\phi_i(x)\cdot(x - \rep(\patch_i)) 
\end{align*}
Substituting for $\Gradient\psi_i(x)$ from Eq.~\eqref{eq:grad_wi}, we obtain the desired form in Eq.~\eqref{eq:grad_sDF_exp}.

Next, we proceed to bound the magnitude of the gradient $\|\Gradient \sDFSq\|$. A key observation we make use  of is the following consequence of the partition of unity construction
\[
    \sum_i \phi_i(x) 
        ~ = ~ 1 
    ~\implies~ 
    \sum_i \Gradient \phi_i(x) 
        ~ = ~ 0.
\]
In addition, we note that the required bound on each local approximation ensures that
\[
    \localF_i(x) 
        ~ = ~ \bigg( 1 + \epsilon_i(x) \bigg) \cdot \DFSq(x), \qquad \text{ where } 0 \leq \epsilon_i(x) \leq \eps, \text{ for all $i \in I$}.
\]
Recalling the definition of the smooth distance approximation from Eq.~\eqref{eq:pu}
\[
    \sDFSq(x) 
        ~ = ~ \sum_i \phi_i(x) \cdot \localF_i(x),
\]
summing over at most $\depth$ non-zero terms. We proceed to compute the gradient as
\begin{align*}
    \Gradient \sDFSq(x) 
        & ~ = ~ \sum_i \localF_i(x) \cdot \Gradient \phi_i(x) + \sum_i \phi_i(x) \cdot \Gradient \localF_i(x) \\
        & ~ = ~ \sum_i \bigg(1 + \epsilon_i(x)\bigg)\cdot\DFSq(x) \cdot \Gradient \phi_i(x) + \sum_i \phi_i(x) \cdot \Gradient \localF_i(x) \\
        & ~ = ~ \DFSq(x) \cdot \bigg(\cancel{\sum_i \Gradient \phi_i(x)} + \sum_i\epsilon_i(x)\cdot\Gradient \phi_i(x)\bigg) + \sum_i \phi_i(x) \cdot \Gradient \localF_i(x) \\
        & ~ = ~ \DFSq(x) \cdot \sum_i\epsilon_i(x)\cdot\Gradient \phi_i(x) + \sum_i \phi_i(x) \cdot \Gradient \localF_i(x).
\end{align*}
We can then bound the magnitude of the gradient as
\begin{align*}
    \|\Gradient \sDFSq(x)\| 
        & ~ \leq ~ \DFSq(x) \cdot \sum_i\epsilon_i(x)\cdot \|\Gradient \phi_i(x)\| + \sum_i \phi_i(x) \cdot \|\Gradient \localF_i(x)\| \\
        & ~ \leq ~ \eps\cdot\DFSq(x)\cdot\depth\cdot\max_i \|\Gradient \phi_i(x)\| + \sum_i \phi_i(x) \cdot \max_i \|\Gradient \localF_i(x)\| \\
        & ~ \leq ~ \eps\cdot\DF^2(x)\cdot\depth\cdot O\bigg(\frac{1}{\eps\cdot\DF(x)}\bigg) + \sum_i \phi_i(x) \cdot O\left(\DF(x)\right) \\
        & ~ =    ~ O\left(\DF(x)\right).
\end{align*}
Deferring the analysis of the coefficients $\zeta_i^{(k)}$ to Lemma~\ref{coefficient_gradient_bound.lem}, this concludes the proof.
\end{proof}

\subsubsection{Smooth Distance Hessian.} \label{sec:sDFHessianBound}
In this section, we establish an upper bound on the norm of the Hessian for the smooth distance approximation as stated in the following lemma.

\sDFHessianBound*

\begin{proof}
Recall that
\[
    \| \Hess \sDFSq(x) \| 
        ~ = ~ \max_{\|v\| = \|u\| = 1} | \Gradient_v \Gradient_u \sDFSq(x)|,
\]
where $\Gradient_u$ and $\Gradient_v$ denote the directional derivatives in the directions of $u$ and $v$, respectively. Recalling that $\Gradient_u \sDFSq(x) = \langle u, \Gradient \sDFSq(x) \rangle$ We begin by rewriting the above inequality as
\begin{align*}
    \| \Hess \sDFSq(x) \| 
        & ~ =    ~ \max_{\|v\| = 1} \max_{\|u\| = 1} | \Gradient_v \langle u, \Gradient \sDFSq(x) \rangle| \\
        & ~ =    ~ \max_{\|v\| = 1}\max_{\|u\|=1}\lim_{\delta \to 0} \frac{|\langle u,  \Gradient\sDFSq(x + \delta v) - \Gradient\sDFSq(x)\rangle|}{\delta} \\
        & ~ \leq ~ \max_{\|v\| = 1}\lim_{\delta \to 0} \frac{\|\Gradient\sDFSq(x + \delta v) - \Gradient\sDFSq(x) \|}{\delta}
\end{align*}
To further simplify the notation, we use $y = x + \delta v$ with $y \to x$ when convenient.
\begin{align*}
    \| \Hess \sDFSq(x) \| 
        & ~ \leq ~ \max_{\|v\| = 1} \lim_{\delta \to 0} \frac{\|\Gradient\sDFSq(x + \delta v) - \Gradient\sDFSq(x) \|}{\delta} \\
        & ~ \leq ~ \max_{\|v\| = 1} 2\cdot\sum_i \lim_{y \to x} \frac{\|\phi_i(y)\cdot(y - \rep(\patch_i)) - \phi_i(x)\cdot(x - \rep(\patch_i))\|}{\|y - x\|}  \\
        & \qquad + \sum_i\sum_k \lim_{y \to x} \frac{\|\zeta_i^{(k)}(y)\cdot\Gradient\shape_i^{(k)}(y) - \zeta_i^{(k)}(x)\cdot\Gradient\shape_i^{(k)}(x)\|}{\|y - x\|}\\
        & ~ \leq ~ \max_{\|v\| = 1} 2\cdot\sum_i \lim_{y \to x} \frac{\|(\phi_i(y) - \phi_i(x))\cdot(x - \rep(\patch_i))\| + \|\phi_i(y)\cdot(y - x)\|}{\|y - x\|}\\
        &\qquad + \sum_i\sum_k \lim_{y \to x} \frac{\|(\zeta_i^{(k)}(y) - \zeta_i^{(k)}(x))\cdot\Gradient\shape_i^{(k)}(x)\| + \|\zeta_i^{(k)}(y)\cdot(\Gradient\shape_i^{(k)}(y) - \Gradient\shape_i^{(k)}(x))\|}{\|y - x\|} \\
        & ~ \leq ~ \max_{\|v\| = 1} 2\cdot\sum_i \|x - \rep(\patch_i)\|\cdot\lim_{\delta \to 0} \frac{|\phi_i(x + \delta v) - \phi_i(x)|}{\delta} + |\phi_i(x)|\cdot\lim_{\delta\to 0}\frac{\|\delta v\|}{\delta} \\
        & \qquad + \sum_i\sum_k \|\Gradient\shape_i^{(k)}(x)\|\cdot\lim_{\delta \to 0} \frac{|\zeta_i^{(k)}(x + \delta v) - \zeta_i^{(k)}(x)|}{\delta} \\
        & \qquad\qquad\qquad\qquad + |\zeta_i^{(k)}(x)|\cdot\lim_{\delta\to 0}\frac{\|\Gradient\shape_i^{(k)}(x + \delta v) - \Gradient\shape_i^{(k)}(x)\|}{\delta}\\
        & ~ \leq ~ 2\cdot\sum_i \DF(x)\cdot \|\Gradient \phi_i(x)\| + 1\cdot1 \\
        & \qquad + \sum_i\sum_k O\left(\dfrac{1}{\eps\cdot\DF(x)}\right)\cdot\|\Gradient \zeta_i^{(k)}(x)\| + \left(\eps\cdot\DF^2(x)\right)\cdot\|\Gradient^2 \shape_i^{(k)}(x)\| \\ 
        & ~ \leq ~ 2\cdot\depth\cdot \DF(x)\cdot O\bigg(\frac{1}{\eps\cdot\DF(x)}\bigg) + 1 \\
        & \qquad + \depth\cdot\kappa_\Pi\cdot \bigg[O\bigg(\frac{1}{\eps\cdot\DF(x)}\bigg) \cdot O(\DF(x)) + O\bigg(\eps\cdot\DF^2(x)\bigg) \cdot O\bigg(\frac{1}{\eps^2\cdot\DF^2(x)}\bigg)\bigg] \\
        & ~ \leq ~ O(1/\eps). 
\end{align*}
which completes the proof.
\end{proof}

To finish the analysis above, recall the coefficients $\zeta_i^{(k)}$ from the gradient analysis of the shape function $\shape_i^{(k)}$ in Lemma~\ref{sDF_gradient.lem}. The following lemma bounds their gradients.

\begin{restatable}{lemma}{coefficientGradientBound}\label{coefficient_gradient_bound.lem}
For any compatible shape function $\shape_i^{(k)}$, for $i \in I$ and $k \in k_i$, we have $\|\Gradient \zeta_i^{(k)}(x)\| = O(\DF(x))$.
\end{restatable}

\begin{proof}
Recalling the definition from Eq.~\eqref{eq:shorthand_coefficient}
\[
    \zeta_i^{(k)}(x) 
        ~ = ~ \frac{2}{\Psi(x)}\cdot\frac{\psi_i(x)\cdot \shape_i^{(k)}(x)}{(\shape_i^{(k)}(x)^2-1)^2}\cdot\left(\sDFSq(x) - \localF_i(x)\right),
\]
we proceed to compute and bound the gradient as follows.
\begin{align*}
    \Gradient \zeta_i^{(k)}(x)
        & ~ = ~ \bigg(\frac{-2}{\Psi^2(x)}\cdot\sum_j \Gradient \psi_j(x) \bigg)\cdot\frac{\psi(x)\cdot \shape_i^{(k)}(x)}{(\shape_i^{(k)}(x)^2 - 1)^2}\cdot\bigg(\sDFSq(x) - \localF_i(x)\bigg) \\
        &\quad +\frac{\sDFSq(x) - \localF_i(x)}{\Psi(x)}\cdot\bigg[\frac{\psi(x)\cdot (1+3\shape_i^{(k)}(x)^2)}{(1 - \shape_i^{(k)}(x)^2)^3}\cdot\Gradient \shape_i^{(k)}(x)  \\
        & \qquad\qquad\qquad\qquad\qquad\qquad~~ + \frac{\shape_i^{(k)}(x)}{(\shape_i^{(k)}(x)^2 - 1)^2}\cdot\Gradient\psi(x)\bigg] \\
        &\quad +\frac{2}{\Psi(x)}\cdot\frac{\psi(x)\cdot \shape_i^{(k)}(x)}{(\shape_i^{(k)}(x)^2 - 1)^2}\cdot\bigg(\Gradient\sDFSq(x) - 2\cdot(x - \rep(\Omega))\bigg). \\
    \|\Gradient \zeta_i^{(k)}(x) \|
        & ~ \leq ~ O(\max_j \|\Gradient \psi_j(x)\|)\cdot \frac{1}{2}(\eps\cdot\DFSq(x)) + \frac{1}{2} O(\DF(x)) \\
        &\quad + O(\eps\cdot\DFSq(x))\cdot\bigg[5\cdot\max_k\|\Gradient \shape_i^{(k)}(x)\| + \bigg|\frac{\psi(x)\cdot \shape_i^{(k)}(x)^2}{(\shape_i^{(k)}(x)^2-1)^2}\bigg| \\
        &\qquad\qquad\qquad\qquad\qquad\qquad~~ + \max_{k'}O\bigg(\bigg|\frac{\psi(x)\cdot \shape_i^{(k)}(x)\cdot \shape_i^{(k')}(x)}{(\shape_i^{(k)}(x)^2-1)^2\cdot(\shape_i^{(k')}(x)-1)^2}\bigg|\bigg)\bigg] \\
        & ~ \leq ~ O(\DF(x)) + O\bigg(\eps\cdot\DFSq(x)\cdot\|\Gradient \shape_i^{(k)}(x)\|\bigg) \\
        & ~ \leq ~ O(\DF(x)), 
\end{align*}
where we used $\shape_i^{(k')}$ to indicate terms with mixed shape functions. In addition, we utilized the following simplifying bounds similar to Eq.~\eqref{eq:grad_bound}.
\begin{align*}
    \left| \frac{\psi(x)\cdot (1+3\cdot \shape_i^{(k)}(x)^2)}{(1 - \shape_i^{(k)}(x)^2)^3} \right| 
        & ~ \leq ~ \left| \frac{\mu(\shape_i^{(k)}(x))\cdot (1+3\cdot \shape_i^{(k)}(x)^2)}{(1 - \shape_i^{(k)}(x)^2)^3} \right| 
          ~ < ~ 5. \\
    \left| \frac{\psi(x)\cdot \shape_i^{(k)}(x)^2}{(\shape_i^{(k)}(x)^2-1)^4} \right| 
        & ~ \leq ~ \left| \frac{\mu(\shape_i^{(k)}(x))\cdot \shape_i^{(k)}(x)^2}{(\shape_i^{(k)}(x)^2-1)^4} \right| 
          ~ < ~ 4. \\
    \bigg|\frac{\psi(x)\cdot \shape_i^{(k)}(x)\cdot \shape_i^{(k')}(x)}{(\shape_i^{(k)}(x)^2-1)^2\cdot(\shape_i^{(k')}(x)-1)^2}\bigg|
        & ~ \leq ~ \bigg|\frac{\mu(\shape_i^{(k)}(x))\cdot \shape_i^{(k)}(x)}{(\shape_i^{(k)}(x)^2-1)^2}\bigg| \cdot \bigg|\frac{\mu(\shape_i^{(k')}(x))\cdot \shape_i^{(k')}(x)}{(\shape_i^{(k')}(x)-1)^2}\bigg| 
          ~ \leq ~ \frac{1}{4}.
\end{align*}
which completes the proof.
\end{proof}

\subsection{Expansion-Containment for Voronoi-based Macbeath Regions.} \label{sec:vor-exp-con}

The purpose of this section is to prove Lemma~\ref{vor-exp-con.lem}, which generalizes expansion-containment to Voronoi-based Macbeath regions. An important property is that the cells of the Voronoi diagram are the projections of facets of a convex polyhedron in $\RE^{d+1}$ that arises through a lifting process. This allows us to relate the Macbeath regions of expanded Voronoi cells to Macbeath regions in this convex polyhedron. (See Clarkson \cite{Cla94} for another approach to approximating Voronoi diagrams through the lifting transformation.) 

We begin by recalling a few properties of the well-known \emph{lifting transformation}~\cite{Boy98, BCK10}. We will represent points in $\RE^{d+1}$ as a pair $(x,w)$, where $x \in \RE^d$ and $w \in \RE$. Define $\Psi$ to be the paraboloid in $\RE^d$ consisting of the points $(x,w)$ where $w = \|x\|^2 = \inner{x}{x}$ (see Figure~\ref{lifting.fig}(a)). Let us think of the $w$-axis as being directed vertically upwards. For any $(x,w) \in \RE^{d+1}$ define $(x,w)^{\downarrow} = x$. Given a set $X \subseteq \RE^{d+1}$, define $X^{\downarrow}$ as the vertical projection of this set onto $\RE^d$. 

For any $p \in \RE^d$, define $h(p) = \{(x,w) : w = \ang{2 x - p, p}\}$. It is straightforward to verify that $h(p)$ is a hyperplane that is tangent to $\Psi$ at the vertical projection of $p$ onto $\Psi$ (see Figure~\ref{lifting.fig}(a)). The relevance of these hyperplanes to nearest-neighbor searching is that the vertical distance between any point $(x,w) \in h(p)$ and $\Psi$ is equal to the squared Euclidean distance between $x$ and $p$.%
\footnote{To see this, observe that the squared distance between $x$ and $p$ is $\ang{x - p, x - p} = \ang{x,x} - 2 \ang{x,p} + \ang{p,p} = \|x\|^2 - \ang{2 x - p, p}$, which is just the vertical distance between $(x,w) = (x, \ang{2 x - p, p})$ and $\Psi$.}
Define $\mathcal{E}(P) = \mathcal{E}$ to be the intersection of the closed upper halfspaces bounded by $h(p)$, for all $p \in P$ (see Figure~\ref{lifting.fig}(a)). Clearly, $\mathcal{E}$ is an (unbounded) convex polyhedron that has $n$ facets, one per site. It follows from our earlier remark about squared distances that the vertical projection of the closed facet of $\mathcal{E}$ bounded by $h(p)$ onto $\RE^d$ is just the Voronoi cell $V(p)$. 

\begin{figure}[htbp]
 \centerline{\includegraphics[scale=0.40]{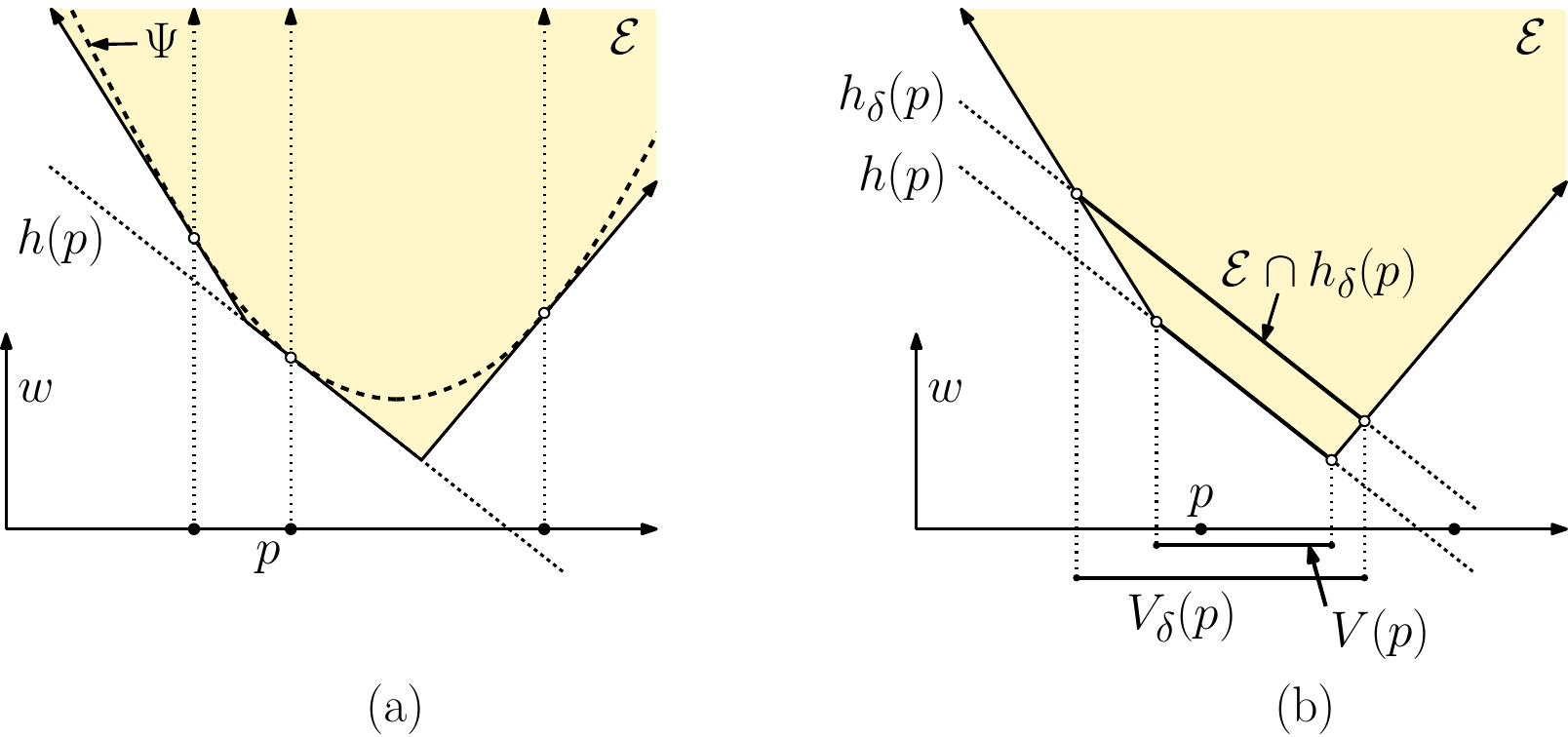}}
 \caption{(a) The lifting transformation and (b) expanded Voronoi cells.} \label{lifting.fig}
\end{figure}

For $\delta \geq 0$, define $h_{\delta}(p)$ to be the translate of $h(p)$ vertically upwards by a distance of $\delta^2$. Clearly, if a point $(x,w)$ lies at the intersection of $h_{\delta}(p)$ and $h_0(p')$, then $\|x - p\|^2 - \delta^2 = \|x - p'\|^2$, or equivalently  $\|x - p\|^2 = \|x - p'\|^2 + \delta^2$. Thus, $x$ lies on the shifted bisector that defines the $\delta$-expanded Voronoi cell $V_{\delta}(p)$. We have the following as a consequence (see Figure~\ref{lifting.fig}(b)).

\begin{lemma} \label{exp-vor-lift.lem}
Given a finite set $P \subset \RE^d$ and scalar $\delta \geq 0$, $V_{\delta}(p) = (\mathcal{E} \cap h_{\delta}(p))^{\downarrow}$.
\end{lemma}

In order to relate the Macbeath regions arising from different Voronoi cells, we lift their defining points onto $\mathcal{E}$ and consider the relation of the associated Macbeath regions in $\RE^{d+1}$. To make this formal, given $x \in \RE^d$, let $p = \nn_P(x)$ and define $x^{\delta} \in \RE^{d+1}$ to be the vertical projection of $x$ onto the shifted hyperplane $h_{\delta}(p)$ (see Figure~\ref{macbeath-related.fig}(a)). The following lemma relates the $d$-dimensional Voronoi-based Macbeath region $M_{\delta}(x)$ to the $(d+1)$-dimensional Macbeath region $M_{\mathcal{E}}(x^{\delta})$. The first part shows that if we slice $M_{\mathcal{E}}(x^{\delta})$ by the hyperplane $h_{\delta}(p)$ (which passes through $x^{\delta}$) its vertical projection is equal to $M_{\delta}(x)$ (see Figure~\ref{macbeath-related.fig}(a)). The second part shows that the vertical projection of the $M_{\mathcal{E}}(x^{\delta})$ is nested between $M_{\delta}(x)$ and its factor-2 expansion.

\begin{figure}[htbp]
 \centerline{\includegraphics[scale=0.40]{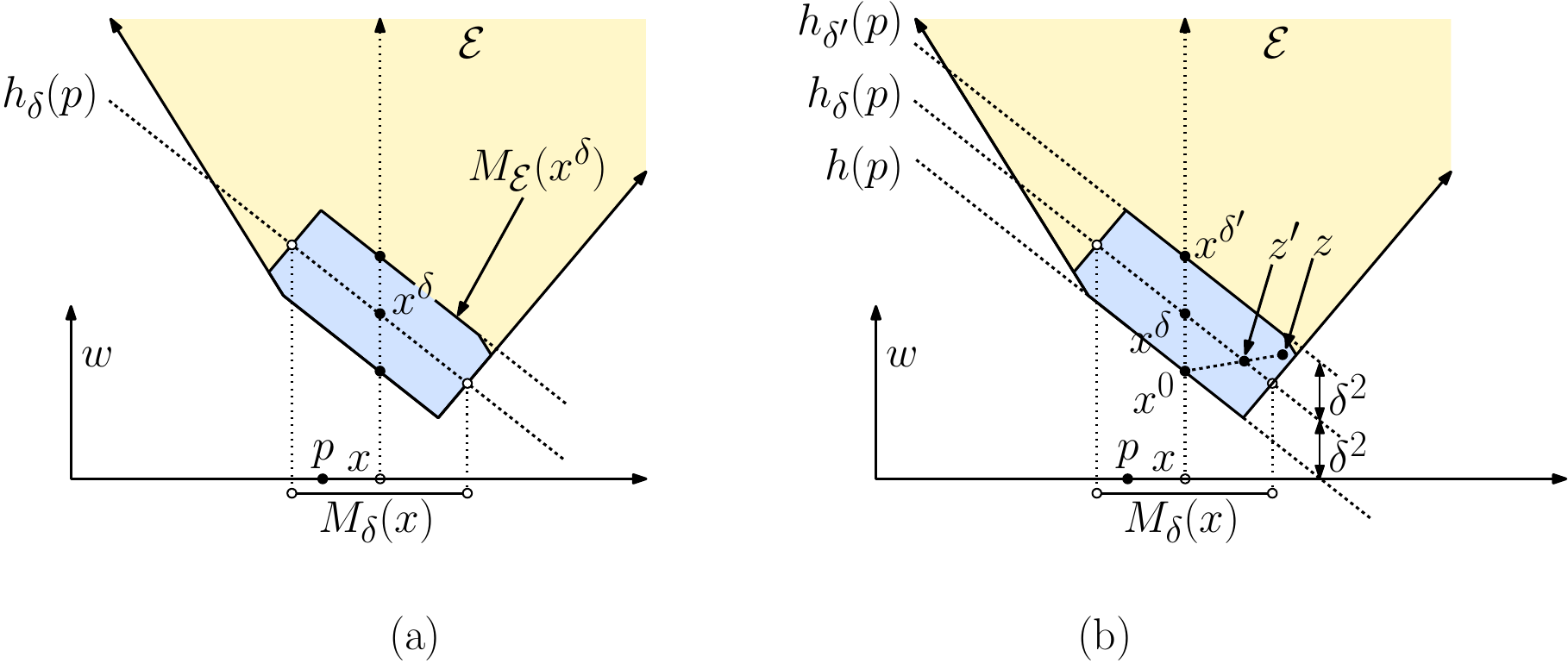}}
 \caption{The relationship between Voronoi-based Macbeath regions and Macbeath regions of $\mathcal{E}$.} \label{macbeath-related.fig}
\end{figure}

\begin{lemma} \label{macbeath-related.lem}
Consider a finite set $P \subset \RE^d$, scalar $\delta \geq 0$, and $x \in \RE^d$. Letting $p = \nn_P(x)$, we have: 
\begin{enumerate}
\item[$(a)$] $(M_{\mathcal{E}}(x^{\delta}) \cap h_{\delta}(p))^{\downarrow} ~=~ M_{\delta}(x)$,

\item[$(b)$] $M_{\delta}(x) ~\subseteq~ M_{\mathcal{E}}(x^{\delta})^{\downarrow} ~\subseteq~ M_{\delta}^2(x)$.
\end{enumerate}
\end{lemma}

\begin{proof}
We start with assertion~(a). By the definition of Macbeath region, 
\[
	M_{\mathcal{E}}(x^{\delta}) \cap h_{\delta}(p)
		~ = ~ [\mathcal{E} \cap (2 x^{\delta} - \mathcal{E})] \cap h_{\delta}(p)
		~ = ~ (\mathcal{E} \cap h_{\delta}(p)) \cap [(2 x^{\delta} - \mathcal{E}) \cap h_{\delta}(p)].
\]
Since $x^{\delta}$ lies on $h_{\delta}(p)$, a point of $\mathcal{E}$ contributes a point to $(2 x^{\delta} - \mathcal{E}) \cap h_{\delta}(p)$ only if this point lies on $h_{\delta}(p)$.%
\footnote{Otherwise, the vector $v$ between $x^{\delta}$ and this point is not parallel to $h_{\delta}(p)$, implying that neither $x^{\delta} + v$ nor $x^{\delta} - v$ lies on $h_{\delta}(p)$.}
Therefore, $(2 x^{\delta} - \mathcal{E}) \cap h_{\delta}(p) = 2 x^{\delta} - (\mathcal{E} \cap h_{\delta}(p))$, and hence
\[
	M_{\mathcal{E}}(x^{\delta}) \cap h_{\delta}(p)
		~ = ~ (\mathcal{E} \cap h_{\delta}(p)) \cap [2 x^{\delta} - (\mathcal{E} \cap h_{\delta}(p))].
\]
By projecting back to $\RE^d$, applying Lemma~\ref{exp-vor-lift.lem}, and observing that $(x^{\delta})^{\downarrow} = x$, we have
\begin{eqnarray*}
	(M_{\mathcal{E}}(x^{\delta}) \cap h_{\delta}(p))^{\downarrow}
		& = & (\mathcal{E} \cap h_{\delta}(p))^{\downarrow} \cap [2 (x^{\delta})^{\downarrow} - (\mathcal{E} \cap h_{\delta}(p))^{\downarrow}] \\
		& = & V_{\delta}(p) \cap (2 x - V_{\delta}(p)),
\end{eqnarray*}
which by the definition of Macbeath regions is $M_{\delta}(x)$, thus establishing (a).

In order to establish (b), let ${\delta'} = \sqrt{2}\delta$, so that ${\delta'}^{\kern+1pt 2} = 2 \delta^2$. The three hyperplanes $h(p) = h_0(p)$, $h_{\delta}(p)$ and $h_{\delta'}(p)$ are parallel to each other and adjacent pairs are separated by a vertical distance of $\delta^2$ (see Figure~\ref{macbeath-related.fig}(b)). The points $x^0$, $x^{\delta}$, and $x^{\delta'}$ lie vertically above $x$ on these hyperplanes, respectively, and are similarly spaced. Because $x^{\delta}$ is midway between $x^0$ and $x^{\delta'}$, we have $x^0 = 2x^{\delta} - x^{\delta'}$.

First, we assert that $M_{\mathcal{E}}(x^{\delta})$ lies entirely within the slab bounded by $h(p)$ and $h_{\delta'}(p)$. This is because $M_{\mathcal{E}}(x^{\delta})$ is centrally symmetric about $x^{\delta}$, which lies midway between these two hyperplanes, and no point of $M_{\mathcal{E}}(x^{\delta})$ can lie below $h(p)$, since this is a bounding hyperplane of $\mathcal{E}$. Second, we assert that $x^0 \in M_{\mathcal{E}}(x^{\delta})$. To see why, observe that $x \in V(p)$, which implies that its vertical projection $x^0$ lies on a facet of $\mathcal{E}$, thus $x^0 \in \mathcal{E}$. Also, the symmetric point about $x^{\delta}$, namely $x^{\delta'}$, also clearly lies within $\mathcal{E}$. Thus, we have $x^0 = 2 x^{\delta} - x^{\delta'} \in 2 x^{\delta} - \mathcal{E}$, which implies that $x_0 \in \mathcal{E} \cap (2 x^{\delta} - \mathcal{E}) = M_{\mathcal{E}}(x^{\delta})$.

Let $z$ be any point of $M_{\mathcal{E}}(x^{\delta})$. We will show that the vertical projection of $z$ onto $\RE^d$ lies within $M_{\delta}^2(x)$. There are two cases, depending on whether $z$ lies above or below $h_{\delta}(p)$. We consider just the former case, because the other one follows by the symmetry of $M_{\mathcal{E}}(x^{\delta})$ about $x^{\delta} \in h_{\delta}(p)$. By our earlier assertion, $z$ lies on or below $h_{\delta'}(p)$. Let $z'$ denote the point where the line segment $z x^0$ intersects $h_{\delta}(p)$ (see Figure~\ref{macbeath-related.fig}(b)). By convexity, $z'$ lies within $M_{\mathcal{E}}(x^{\delta}) \cap h_{\delta}(p)$, which by assertion~(a) projects vertically onto $M_{\delta}(x)$. Clearly, $\|z - x^0\|/\|z' - x^0\| \leq 2$, which implies that $z$ projects vertically onto a point that lies within the factor-2 scaling of $M_{\delta}(x)$ about $x$. By definition, this is $M_{\delta}^2(x)$, thus establishing~(b) and completing the proof.
\end{proof}

\begin{figure}[htbp]
 \centerline{\includegraphics[scale=0.40]{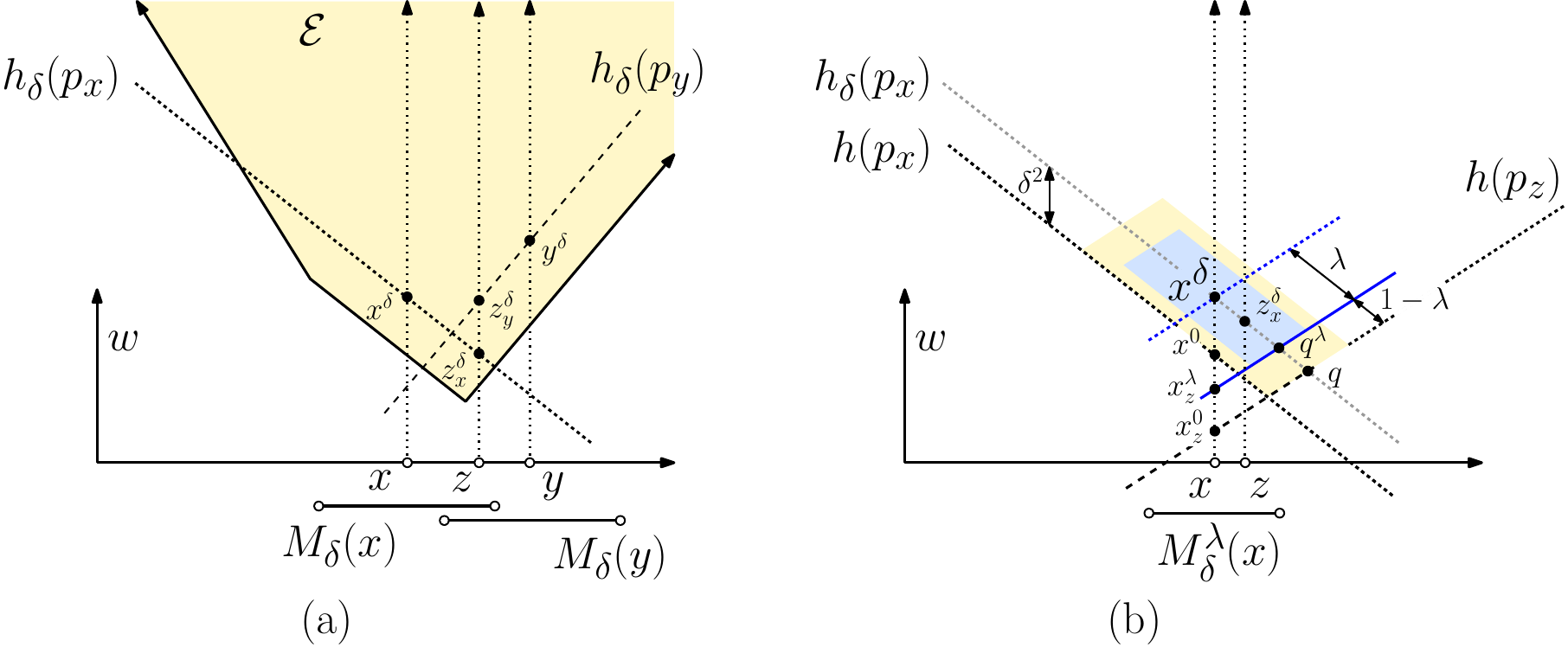}}
 \caption{Expansion-containment for distance-based Macbeath regions via the lifted polytope $\mathcal{E}$.} \label{macbeath-exp-no-bisector.fig}
\end{figure}

We can now present the proof of Lemma~\ref{vor-exp-con.lem}, which we restate below.

\VorExpCon*

\begin{proof}
Fix a point $z \in M_{\delta}^{\lambda}(x) \cap M_{\delta}^{\lambda}(y)$. Let $p_x = \nn(x)$, $p_y = \nn(y)$, and $p_z = \nn(z)$. If $p_x = p_y$, then the result follows directly from the expansion-containment properties of standard Macbeath regions~\cite{AFM17a}. Otherwise, let $z^0$, $z^\delta_x$ and $z^\delta_y$ be the vertical projections of $z$ onto $h(p_z)$, $h_{\delta}(p_x)$, and $h_{\delta}(p_y)$, respectively (see Figure~\ref{macbeath-exp-no-bisector.fig}(a)). By construction, both $z^\delta_x$ and $z^\delta_y$ lie above $z^0$ with a vertical separation of at most $\delta^2$ and at least $(1 - \lambda)\delta^2$.

To see this, let $x^0_z$ be the vertical lifting of $x$ onto $h(p_z)$ and $q = \overrightarrow{x^\delta z^\delta_x} \cap h(p_z)$ (see Figure~\ref{macbeath-exp-no-bisector.fig}(b)). By definition, $x^\delta$ lies above $x^0$ at a distance of $\delta^2$, while $x^0_z$ cannot lie above $x^0$ (for otherwise, $p_z$ would be closer to $x$ than $p_x$). It follows that $\|x^\delta - x^0_z\| \geq \delta^2$. Now, scaling by a factor of $\lambda$ with $x^\delta$ as the center, the point $x^0_z$ is mapped to a point $x^\lambda_0$, which lies above $x^0_z$ at a distance at least $(1-\lambda)\delta^2$. Similarly, the point $q$ is mapped to a point $q^\lambda$. Observing that $\overline{x^0_zq} \parallel \overline{x^\lambda_zq^\lambda}$, it follows that $q^\lambda$ lies above $h(p_z)$ at a distance at least $(1-\lambda)\delta^2$. The same bound on the vertical separation above $h(p_z)$ also applies to the point $z^\delta_x \in \overline{x^\delta q^\lambda}$. A similar argument holds for the vertical separation of $z^\delta_y$ above $h(p_z)$.

Let $z^\ast$ be the point above $z^0$ at a distance of $(1 - \lambda/2)\delta^2$, that is, the midpoint of the interval above $z^0$ containing both $z^\delta_x, z^\delta_y$. By construction, $z^0, z^\delta_x, z^\delta_y \in M_\mathcal{E}(z^\ast)$, since the lifted polytope $\mathcal{E}$ is unbounded above $z^0$. We seek a scaling factor $\varrho\lambda$ about $z^\ast$ such that $z^0$ is mapped to a point $z^{\varrho\lambda}$ with $\|z^\ast - z^{\varrho\lambda}\| \geq \lambda/2\cdot\delta^2$, i.e., the radius of the interval, guaranteeing that $z^\delta_x, z^\delta_y \in M_\mathcal{E}^{\varrho\lambda}(z^\ast)$. Writing $\lambda/2\cdot\delta^2 \leq \varrho\lambda\cdot(1 - \lambda)\delta^2$, we have $\varrho \geq \frac{1}{2(1-\lambda)}$.

Hence, $M^{\varrho\lambda}_\mathcal{E}(x^\delta) \cap M^{\varrho\lambda}_\mathcal{E}(z^\delta_\ast) \neq \emptyset$ and $M^{\varrho\lambda}_\mathcal{E}(y^\delta) \cap M^{\varrho\lambda}_\mathcal{E}(z^\delta_\ast) \neq \emptyset$. Applying standard expansion-containment, we obtain
\[
	M_{\delta}^{\alpha\lambda}(y) 
		~ \subseteq ~ M_{\mathcal{E}}^{\alpha\lambda}(y^{\delta})^{\downarrow} 
		~ \subseteq ~ M_{\mathcal{E}}^{\alpha\varrho\lambda}(y^{\delta})^{\downarrow} 
		~ \subseteq ~ M_{\mathcal{E}}^{\beta'\lambda}(z^{\ast})^{\downarrow} 
		~ \subseteq ~ M_{\mathcal{E}}^{\beta\lambda/2}(x)^{\downarrow}
		~ \subseteq ~ M_{\delta}^{\beta\lambda}(x).
\]
where $\varrho = \frac{1}{2(1 - \lambda)}$, $\beta' = \frac{2 + \alpha\varrho(1+\lambda)}{1 - \lambda}$ and $\beta = 2 \frac{2 + \beta'(1+\lambda)}{1 - \lambda} = \frac{4}{(1-\lambda)^2} + \frac{(1+\lambda)^2}{2(1-\lambda)^3}\alpha$. This establishes Lemma~\ref{vor-exp-con.lem}.
\end{proof}

\subsection{Bounding the Number of Distance-based Macbeath Regions.} \label{sec:vor-ecc}

In this section, we present a proof of Lemma~\ref{vor-ecc.lem}, which restate below.

\rsVorECC*

We begin by recalling a useful lemma, proved by Arya {\etal}~\cite[Lemma~{3.2}]{AFM17c}, which bounds the number of Macbeath regions close to the boundary of the convex body. Before presenting the lemma, we recall some basic concepts. For $0 < \kappa \leq 1$, we say that $K$ is in \emph{$\kappa$-canonical form} if $K$ can be sandwiched between two Euclidean balls centered at the origin, one of diameter $\kappa$ and the other of diameter $1$ (see Figure~\ref{cap.fig}(a)). Given a convex body $K$, a \emph{cap} $C$ is defined to be the nonempty intersection of the convex body $K$ with a halfspace (see Figure~\ref{cap.fig}(b)). Letting $h$ denote the hyperplane bounding this halfspace, define $C$'s \emph{base} to be $h \cap K$, and define $C$'s \emph{width} to be the distance between $h$ and this supporting hyperplane.

\begin{lemma}[Arya {\etal}~\cite{AFM17c}] \label{ecc.lem}
Let $K \subset \RE^d$ be a convex body in $\kappa$-canonical form. Let $\Delta \leq \kappa/12 d$ be a real parameter, and let $\lambda \leq 1/5$ be a constant. Let $\mathcal{C}$ be a set of caps, whose widths lie between $\Delta$ and $2\Delta$, such that the Macbeath regions $M^{\lambda}_K(x)$ centered at the centroids $x$ of the bases of these caps are disjoint. Then $|\mathcal{C}| = O(1/\Delta^{(d-1)/2})$ (see Figure~\ref{cap.fig}(c)).
\end{lemma}

\begin{figure}[htbp]
  \centerline{\includegraphics[scale=0.40]{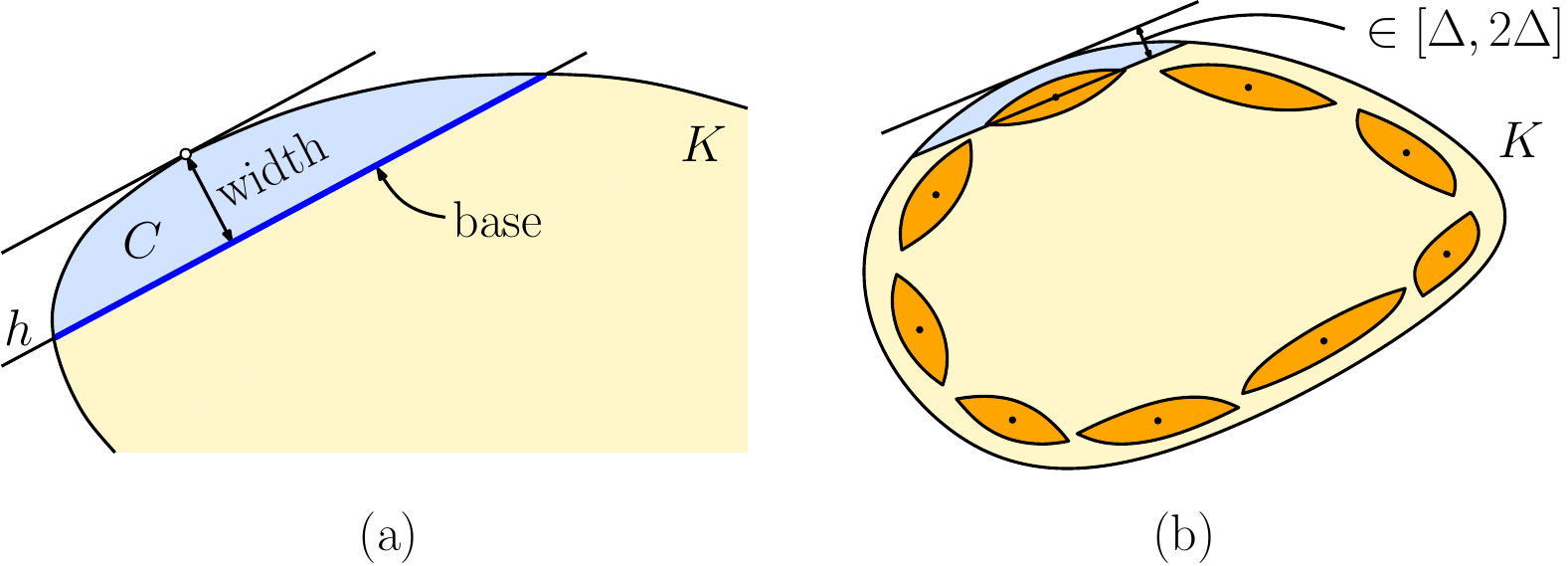}}
  \caption{(a) Canonical form, (b) cap definitions, and (c) the cap cover.} \label{cap.fig}
\end{figure}

To prove Lemma~\ref{vor-ecc.lem}, we consider two cases, depending on which set is the inner set and which is the outer set. We may assume that $\lambda$ is smaller than any fixed constant, since decreasing the value of $\lambda$ can only increase the worst-case size of the point set. Let $r$ denote the radius of ball $b$ containing the inner set. Without loss of generality, let us translate space so this ball is centered at the origin.

First, consider the case where $P$ is the inner set and $w$ the outer set. That is, $P$ is contained within a ball $b$ and $w$ lies outside of $2 b$ but the points of $X$ lie within $w \cap \gamma b$ (recall Figure~\ref{evd-setup.fig}(b)). We apply the lifting transformation (recall Section~\ref{sec:vor-exp-con}) to the points of $P$ resulting in the polyhedron $\mathcal{E} = \mathcal{E}(P)$. Let us restrict this unbounded polyhedron by intersecting it with a vertical cylinder whose horizontal cross-section is $2 \gamma b$. (For the purposes of containing the points of $X$, it suffices that the cross-section be just $\gamma b$, but we increase the radius to allow the Macbeath regions room to grow.) Let $\mathcal{E}_{\gamma}$ denote the resulting convex body. Clearly, $\mathcal{E}_{\gamma}$ has a horizontal diameter of at most $4 \gamma r$. Since each point $p \in P$ is in the inner set, and hence lies within distance $r$ of the origin, it follows that the absolute coordinates of the gradient of the lifted tangent hyperplane $h(p)$ are at most $2 r$. This implies that $\mathcal{E}_{\gamma}$ has a vertical extent of at most $2 \sqrt{d} r (4 \gamma r) = O(\gamma r^2)$. 

The Voronoi-based Macbeath regions are equivalent up to constant factors to the Macbeath regions of $\mathcal{E}$ that lie at distance $\Theta(\delta^2)$ from the lower portion of the boundary $\mathcal{E}$. More formally, define $\lambda' = \lambda/2\sqrt{d}$. For any $x \in X$, Eq.~\eqref{eq:john-thm} implies that $E^{\lambda}_{\delta}(x) \supseteq M^{\lambda/\sqrt{d}}_{\delta}(x)$. Also, by Lemma~\ref{macbeath-related.lem}(b), we have $M^{\lambda/\sqrt{d}}_{\delta}(x) \supseteq M^{\lambda/2\sqrt{d}}_{\mathcal{E}}(x^{\delta})^{\downarrow}$. Therefore,
\[
	E^{\lambda}_{\delta}(x)
		~ \supseteq ~ M^{\lambda/\sqrt{d}}_{\delta}(x)
		~ \supseteq ~ M^{\lambda/2\sqrt{d}}_{\mathcal{E}}(x^{\delta})^{\downarrow}
		~     =     ~ M^{\lambda'}_{\mathcal{E}}(x^{\delta})^{\downarrow}.
\]
Because the vertical projections are disjoint, it follows that the Macbeath regions $M^{\lambda'}_{\mathcal{E}}(x^{\delta})$ which give rise to these projections are also disjoint. In summary, each point $x \in X$ is in one-to-one correspondence with a lifted point $x^{\delta}$ at vertical distance $\delta^2$ from the lower boundary of $\mathcal{E}$ such that the $\lambda'$-scaled Macbeath regions at these points are pairwise disjoint. Because all the points of $X$ lie within $\gamma b$, it follows that the lifted points $x^{\delta}$ all lie within $\mathcal{E}_{\gamma}$. 

In order to apply Lemma~\ref{ecc.lem}, we need the body to be in canonical form. For any suitably small constant $\kappa$ (depending on the dimension), we can map $\mathcal{E}_{\gamma}$ into $\kappa$-canonical form by scaling each of its horizontal components by $\Theta(1/\gamma r)$ and each of its vertical components by $\Theta(1/\gamma r^2)$. Let $\mathcal{E}^*_{\gamma}$ denote the resulting body. The vertical displacement of each hyperplane $h_{\delta}(p)$ prior to this scaling was $\delta^2$, and after scaling it is $\Delta = \Theta(\delta^2/(\gamma r^2)) = \Theta(\gamma \eps)$. After scaling, the absolute coordinates of the gradients are reduced to $O(r (\gamma r)/(\gamma r^2)) = O(1)$, and hence orthogonal displacements and vertical displacements from the lower boundary of $\mathcal{E}_{\gamma}$ are equivalent to within constant factors.

The linear transformation that scales $\mathcal{E}_{\gamma}$ into $\mathcal{E}^*_{\gamma}$ does not affect disjointness, but now all the points are at distance $\Delta  = \Theta(\gamma \eps)$ from the boundary of $\mathcal{E}^*_{\gamma}$. By applying Lemma~\ref{ecc.lem} in dimension $d+1$ (due to lifting), where $\mathcal{E}^*_{\gamma}$ plays the role of $K$ and $\lambda'$ plays the role of $\lambda$, we conclude that $|X| = O(1/\Delta^{((d+1)-1)/2} = O(1/(\gamma\eps))^{d/2})$, as desired.

Next, consider the case where $w$ is the inner set and $P$ the outer set (see Figure~\ref{evd-setup.fig}(a)). Recall that we assume that the ball $b$ of radius $r$ containing $w$ is centered at the origin. Separation implies that all the points of $P$ lie outside of $2 b$. Again, we apply the lifting transformation to the points of $P$, letting $\mathcal{E} = \mathcal{E}(P)$ denote the resulting polyhedron. This time we restrict this polyhedron to the vertical cylinder whose horizontal cross-section is $2 b$. (As before, the points of $X$ lie within $b$, but we provide additional space for the Macbeath regions to grow into.) Let $\mathcal{E}_{\gamma}$ denote the resulting body. Clearly, it has a horizontal extent $4 r$. Let $p = \nn(x)$ denote the nearest point of $P$ to the origin $O$. We may assume that $\|p\| = \|p - O\| = O(\gamma r)$, since the bound only improves as the points are farther away. Since all the query points lie within distance $r \sqrt{d}$ of the origin, by the triangle inequality it follows that for any point to serve as the nearest neighbor of some query point in $w$, its distance from the origin can be at most $\|p\| + r \sqrt{d}$. Since each relevant point $P$ lies within distance $O(\gamma r)$ of the origin, it follows that the absolute coordinates of the gradient of the lifted tangent hyperplane $h(p)$ are $O(\gamma r)$, and since $\mathcal{E}_{\gamma}$ covers a region of diameter $O(r)$ about the origin, it follows that $\mathcal{E}_{\gamma}$ has a vertical extent $O(\gamma r^2)$. 

As in the previous case, we can map $\mathcal{E}$ into $\kappa$-canonical form by scaling each of its horizontal components by $O(1/r)$ and its vertical component by $O(1/\gamma r^2)$. The vertical displacement of each hyperplane $h_{\delta}(p)$ prior to this scaling was $\delta^2$. After scaling, it is $O(\delta^2/\gamma r^2)) = O(\gamma \eps r)$. After scaling, the absolute coordinates of the gradients are reduced to $O((\gamma r) r/(\gamma r^2)) = O(1)$. The remainder of the analysis follows in the same manner as the previous case.

\subsection{EVD Height.} \label{sec:evd-height}

In this section, we present a proof of Lemma~\ref{root-level.lem}, which shows that the number of levels in the EVD structure grows as the logarithm of $1/\eps$, irrespective of the size of $P$. This follows from the fact that the parameter $\delta_i$, which controls the expansion of the Voronoi cells, increases at a rate that is exponentially larger than the diameter of the region of interest.

\rsRootLevel*

\begin{proof}
We consider two cases, depending on which is the inner set and which is the outer set. First, suppose that $w$ is the inner set, and $P$ is the outer set. Let $x$ be $w$'s center, let $p = \nn(x)$, and let $r = \radius(b)$ where $b$ is the ball containing $P$. Let $i$ be the smallest integer such that $\|p - x\| \geq \gamma_i r$. By concentric separation, $\|p - x\| \geq 2 \gamma_0 r = \gamma_1 r$, and so $i \geq 1$. By the triangle inequality, any other point $p'$ that is the nearest neighbor of a point $q \in w$ satisfies $\|p' - x\| \leq \|p - x\| + \diam(w) \leq 2 \gamma_i r$. Therefore, $\|p' - p\| \leq 4 \gamma_i r$. Recalling the remarks at the start of Section~\ref{sec:approx-voronoi}, at level $i$ of the $\eps$-EVD this bisector is shifted by distance $\delta_i^2 / (2 \|p' - p\|)$. By applying the definition of $\delta_i$, we find that any bisecting hyperplane that might contribute to $V_{\delta_i}(p)$ in $w$ is shifted away by a distance of at least
\[
    \frac{\delta_i^2}{2 \|p' - p\|}
		~  =  ~ \frac{(\gamma_i \sqrt{\eps} r/2)^2}{2 \|p' - p\|}
        ~ \geq ~ \frac{\gamma^2_i \eps r^2}{32 \gamma_i r}
        ~  =  ~ \frac{2^i \eps}{32} \cdot \radius(b).
\]
By setting $i = \ell = O(\log (1/\eps))$, it is possible to make the shift distance larger than $b$'s radius by any constant factor. It follows that the Macbeath ellipsoid $E'_{\ell}(x)$ fully contains $b$ (and hence contains $w$). This implies that $|X_{\ell}| = 1$, and so the construction stops at this level.

Second, suppose that $P$ is the inner set and $w$ is the outer set. As observed earlier, we need only consider the portion of $w$ that lies within $(3/\eps)b$. Let $x$ be any point within $w \cap (3/\eps)b$, and let $p = \nn(x)$. For any point $p' \in P \setminus \{p\}$, we have $\|p' - p\| \leq 2 \gamma_0 r = 2 r$. Therefore, at level $i$ of the $\eps$-EVD, the bisector between $p$ and $p'$ is shifted away from the true bisector by a distance of at least
\[
    \frac{\delta_i^2}{2 \|p' - p\|} 
        ~ \geq ~ \frac{(\gamma_i \sqrt{\eps} r/2)^2}{4 r}
        ~  =  ~ \frac{\gamma^2_i \eps r^2}{16 r}
        ~  =  ~ \frac{4^i \eps}{16} \cdot \radius(b).
\]
By setting $i = \ell = O(\log (1/\eps))$, it is possible to make the shift distance larger than $(3/\eps) \cdot \radius(b)$ by any constant factor. It follows that the Macbeath ellipsoid $E'_{\ell}(x)$ contains $(3/\eps)b$, implying that $|X_{\ell}| = 1$, and the construction stops at this level.
\end{proof}

\subsection{Storage Bounds and the Number of Children in the EVD.} \label{sec:ann-child-bound}

In this section, we first provide a proof of Lemma~\ref{ann-child-bound.lem}. The proof makes use of the following technical result, which states that if the expansion factor $\delta$ of a convex body is doubled, the sizes of the Macbeath regions scale in the expected manner. We omit the proof, but it is a straightforward consequence of the linear behavior under scaling per the definition of Macbeath regions.

\begin{lemma} \label{delta-containment.lem}
Given a convex body $K$, $x \in K$, and any positive reals $\delta$ and $\lambda_0$:
\begin{enumerate} \setlength{\itemsep}{3pt}\setlength{\parsep}{2pt}%
\item[$(a)$] $M^{\lambda}_{K_{\delta}}(x) \subseteq M^{\lambda}_{K_{2\delta}}(x) \subseteq M^{2 \lambda}_{K_{\delta}}(x)$,

\item[$(b)$] $\vol(M^{\lambda}_{K_{\delta}}(x)) \geq \vol(M^{\lambda}_{K_{2\delta}}(x))/2^d$.
\end{enumerate}
\end{lemma}

\rsChildBound*

\begin{proof}
Consider any node $x \in X_i$ at level $i \geq 1$ of the DAG. By definition, its children are associated with the points $y \in X_{i-1}$ such that $E'_{\delta_i}(x) \cap E'_{\delta_{i-1}}(y) \neq \emptyset$. By containment, it follows that $M'_{\delta_i}(x) \cap M'_{\delta_{i-1}}(y) \neq \emptyset$. By the first inclusion of Lemma~\ref{delta-containment.lem}, $M'_{\delta_{i-1}}(y) \subseteq M'_{\delta_i}(y)$, and we have $M'_{\delta_i}(x) \cap M'_{\delta_i}(y) \neq \emptyset$. Next, by applying Lemma~\ref{vor-exp-con.lem} we obtain $M'_{\delta_i}(x) \subseteq M^7_{\delta_i}(y)$.

Recalling that the ellipsoids $E''_{\delta_{i-1}}(y)$ for $y \in X_{i-1}$ are disjoint, by invoking John's Theorem and the corollary to Lemma~\ref{delta-containment.lem}.
\begin{align*}
	\vol(E''_{\delta_{i-1}}(y))
        &~ \geq ~ \vol\big( M^{1/\sqrt{d}}_{\delta_{i-1}}(y) \big)
		 ~ \geq ~ \left(\frac{1}{2}\right)^d \vol\big(M^{1/\sqrt{d}}_{\delta_i}(y)\big)
		 ~ = ~ \left(\frac{1}{14\sqrt{d}}\right)^d \vol(M^{7}_{\delta_i}(y)) \\
		&~ \geq ~ \left(\frac{1}{14\sqrt{d}}\right)^d \vol(M'_{\delta_i}(x))
		 ~ \geq ~ \left(\frac{1}{14\sqrt{d}}\right)^d \vol(E'_{i}(x)).
\end{align*}
Thus, by a packing argument the number of children is at most $\big(14\sqrt{d}\big)^d = O(1)$.
\end{proof}

Next, we bound the total space by presenting a proof of Lemma~\ref{ann-space-bound.lem}. This is a consequence of Lemma~\ref{vor-ecc.lem} and the fact that the space is dominated by the number of nodes at the leaf level.

\rsAnnSpaceBound*

\begin{proof}
By Lemma~\ref{ann-child-bound.lem}, the out-degree of each node of the EVD is bounded, and hence it suffices to bound the number of nodes in the tree, or equivalently, $\sum_i |X_i|$. Let $b$ be the ball containing the inner set. Recall that each set $X_i$ is a maximal set of points lying within $w$ such that the ellipsoids $E^{\lambda_1}_{\delta_i}(x)$ are pairwise disjoint, where $\lambda_1 = 1/(16 \sqrt{d} + 1)$ and $\delta_i = \gamma_i \sqrt{\eps} \cdot \radius(b) /2$.

Suppose first that $P$ is the inner set and $w$ is the outer set. We have $\delta_i = 2^{i-1} \sqrt{\eps} \cdot \radius(b)$. By Lemma~\ref{vor-ecc.lem}, we have $|X_i| \leq c /(2^i \eps)^{d/2}$, for some constant $c$. Therefore, for some constant $c$ the total number of points of all levels is 
\begin{align*}
    \sum_{i=1}^{\ell} |X_i|
        & \leq ~ c \sum_{i=1}^{\ell} \left( \frac{1}{2^i \eps} \right)^{d/2}
        ~ \leq ~ c \left( \frac{1}{\eps} \right)^{d/2} \sum_{i=1}^{\ell} \left( \frac{1}{2^i} \right)^{d/2}
        ~ \leq ~ 2 c \left( \frac{1}{\eps} \right)^{d/2}.
\end{align*}
which is asymptotically $O(1/\eps^{d/2})$, as desired. The other case is similar.
\end{proof}

\subsection{AVD Construction and Bounds for Blending.} \label{sec:avd-blending}

We start by proving the basic separation properties of the AVD cells~\cite{AMM09a} based on the balanced box decomposition (BBD) tree~\cite{ArM00}

\separationProps*

Our proof will make use of the concept of a \emph{well-separated pair decomposition} (WSPD) of a point set. For definition and properties, see Callahan and Kosaraju~\cite{CaK95} or Har-Peled~\cite[Chapter 3]{Har11a}. Before presenting the proof, we begin by establishing a few useful inequalities involving WSPDs. Given a positive parameter $\sigma$, a pair of point sets $X$ and $Y$ is said to be \emph{$\sigma$-well-separated} if each set can be enclosed in a ball of some radius $r$ such that these balls are at distance at least $\sigma r$ from each other. Given a well-separated pair $(X,Y)$ consider the two balls enclosing these sets and the line segment joining their centers. Define the \emph{length} and \emph{midpoint} of the well-separated pair to be length of this line segment and its midpoint, respectively.

\begin{lemma} \label{wspd-props.lem}
Given $\sigma > 4$, consider a $\sigma$-well-separated pair $(X,Y)$ of length $\ell$ and center $z$. For any $x \in X$ and $y \in Y$, $\ell/2 < \|y - x\| < 3\ell/2$, and $\|x - z\| < 3\ell/4$.
\end{lemma}

Our BBD tree involves a quadtree-like subdivision of the unit hypercube $[0,1]^d$. We first scale and translate the point set $P$ so it lies within a ball of sufficiently small ($\Theta(\eps)$) radius centered within this hypercube so that for any query point lying outside the hypercube, any point of $P$ may be taken as its $\eps$-ANN. Next, we compute a $\sigma$-WSPD for $P$ for any $\sigma > 4$. As shown by Callahan and Kosaraju~\cite{CaK95}, this can be done in $O(n \log n)$ time, resulting in $O(n)$ pairs. 

For each well-separated pair $(X,Y)$ of the WSPD, let $\ell$ and $z$ denote its length and center, respectively. Let $c_1$ and $c_2$ be two constants to be specified later. We generate the set of all quadtree boxes of diameter $\ell/c_1$ (rounded down to the next valid quadtree box size) that overlap a ball of radius $c_2 \ell$ centered at $z$. Observe that for each such box, there exists a point of $X \cup Y$ whose distance from the box is at most a constant factor larger than the box's diameter. The number of boxes generated per well-separated pair is $O((c_1 c_2)^d) = O(1)$, and hence there are altogether $O(n)$ boxes. We compute a BBD tree whose leaf-level subdivision is a refinement of these boxes. Arya {\etal} \cite{AMM09a} showed that such a tree can be constructed in $O(n \log n)$ time, and it is possible to determine the leaf cell containing a query point in $O(\log n)$ time. 

We will show that every leaf cell $w$ of the resulting BBD tree satisfies at least one of the following separation properties, where the values of the constant parameters $\alpha \geq 3$ and $\beta \geq 2$ are to be specified later:
\begin{enumerate}
\item[$(a)$] $P \cap \alpha b_w = \emptyset$, implying that $w$ is concentrically $\alpha$-separated from $P$ (see Figure~\ref{separation.fig}(a)).

\item[$(b)$] $\alpha b_w$ contains exactly one point of $P$, which lies within $w$. This point is the nearest neighbor of any query point in $w$ (see Figure~\ref{separation.fig}(b)).

\item[$(c)$] There exists a ball $b'_w$ contained within $\alpha b_w$ that is concentrically $\beta$-separated from $w$, and the nearest neighbor of any point within $w$ lies within $P \cap b'_w$ (see Figure~\ref{separation.fig}(c)).
\end{enumerate}

\noindent
Let $x \in P$ be the nearest neighbor of $w$'s center. We consider three cases:
\begin{itemize}
\item If $x \notin \alpha b_w$, then (a) holds.

\item Suppose that $x \in \alpha b_w$, but $x$ is the only point of $P$ lying within $(\alpha+2) b_w$. If $x \in w$, then (b) holds. Otherwise, we will show that (c) holds. This is because a single point lying outside of $w$ is trivially concentrically $\beta$-separated from $w$, and the distance of any query point in $w$ to $x$ is at most $\alpha + 1$ times the radius of $b_w$ while the distance to its next closest point of $P$, which lies outside of $(\alpha + 2) b_w$, is strictly greater than $(\alpha + 2 - 1) = \alpha + 1$ times this radius.

\item Otherwise, there is at least one more point of $P$, other than $x$, within $(\alpha + 2) b_w$. We will show that property (c) holds. Let $y$ be the farthest such point from $x$. Let $b'_w$ be the ball of radius $\|y - x\|$ centered at $x$. Clearly, all the points of $P$ that lie within $(\alpha + 2) b_w$ lie within $b'_w$. Consider the pair $(X,Y)$ of the WSPD such that $x \in X$ and $y \in Y$. Let $\ell$ and $z$ denote its length and center, respectively. Since both $x$ and $y$ lie within $(\alpha + 2) b_w$, $\|y - x\| \leq 2\cdot(\alpha + 2) r_w$. We assert that by selecting $c_1 > 2\cdot(\alpha + 2)$, then $\dist(z,w) \geq c_2 \ell$. If not, $w$ overlaps the ball of radius $c_2 \ell$ centered at $z$, and by the above construction and Lemma~\ref{wspd-props.lem}(a) it follows that 
\[
	r_w ~  =   ~ \frac{1}{2} \diam(b_w)
		~ \leq ~ \frac{1}{2} \frac{\ell}{c_1}
		~ \leq ~ \frac{\|y - x\|}{c_1}
		~  <   ~ \frac{2\cdot(\alpha + 2)\cdot r_w}{c_1}
		~  <   ~ r_w,
\]
yielding a contradiction. By combining this bound on $\dist(z,w)$ with the triangle inequality and Lemma~\ref{wspd-props.lem}, we have 
\[
	\dist(x,w) 
		~ \geq ~ \dist(z,w) - \|x - z\| 
		~  >   ~ c_2 \ell - \frac{3 \ell}{4}
		~  >   ~ (c_2 - 1) \ell 
		~  >   ~ \frac{2(c_2 - 1)}{3} \|y - x\|.
\]
Since $\|y - x\|$ is the radius of $b'_w$, it follows that by setting $c_2 \geq \frac{3}{2}\beta + 1$, $b'_w$ is concentrically $\beta$-separated from $w$, as desired. As observed before, with $x \in \alpha b_w$, the nearest neighbor of any query point in $w$ cannot lie outside $(\alpha + 2) b_w$, and must therefore lie within $b'_w$.
\end{itemize} 

Next, we bound the number of AVD cells we need to consider for blending. Upon receiving the query point $q$, the standard AVD locates the leaf cell $w^\ast$ containing $q$ in $O(\log n)$ time. In contrast, blending computations require that we allow AVD cells $\{w\}$ to overlap by working instead with their expanded enclosing balls $\{2 b_w\}$ instead. Hence, in order to process the query $q$, we need to enumerate all AVD cells $w$ whenever $q \in 2b_w$, which is a standard quadtree operation. In what follows, we derive an upper bound on the number of such cells, establishing the following lemma.

\avdLookup*

We take the above analysis a step further, and consider well-separated pairs (WSPs) involving points \emph{outside} $\alpha b_w$. We are particularly interested in AVD cells generated by such WSPs and overlapping $\gamma b_w$, for a constant $\gamma \in (1, \alpha)$. Fix one such WSP, and let $z$ and $\ell$ denote its center and length, respectively. Now, the ball of radius $c_2 \ell$ centered at $z$ cannot completely lie within $\alpha b_w$, as it necessarily includes the points outside $\alpha b_w$. Recalling that such a ball generates a set of AVD cells as quadtree boxes of diameter $\ell/c_1$, we have
\[
    \frac{\ell}{c_1} + 2\cdot c_2 \cdot \ell 
        ~ \geq ~ (\alpha - \gamma) r_w 
        ~ \implies ~ 
    \ell 
        ~ \geq ~ \frac{\alpha - \gamma}{1/c_1 + 2\cdot c_2}r_w .\]

\begin{figure}[htbp]
  \centerline{\includegraphics[scale=0.40]{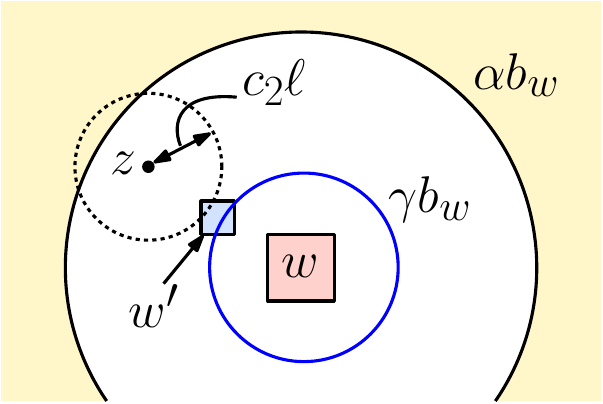}}
  \caption{Bounding AVD refinement due to outer points in the neighborhood of a leaf cell $w$.}
  \label{blend_bound.fig}
\end{figure}

It follows that the smallest such AVD cell $w'$ has diameter
\[
    \diam(w') 
        ~ =    ~ \frac{\ell}{c_1} 
        ~ \geq ~ \frac{\alpha - \gamma}{1 + 2\cdot c_1 \cdot c_2}r_w 
        ~ =    ~ \frac{\alpha - \gamma}{2 + 4\cdot c_1 \cdot c_2}\diam(w),
\]
see Figure~\ref{blend_bound.fig}. Hence,
\[
    \frac{\diam(w)}{\diam(w')} 
        ~ \leq ~ \frac{2 + 4\cdot c_1 \cdot c_2}{\alpha - \gamma} 
        ~ := ~ \Delta.
\]
Observe that whenever $\diam(w_1) \leq \diam(w_2)$, then
\[
    w_2 \cap \gamma b_{w_1} 
        ~ \neq ~ \emptyset 
        ~ \implies ~
    w_1 \cap \gamma b_{w_2} 
        ~ \neq ~ \emptyset.
\]
Hence, for any pair of AVD cells $w_1$ and $w_2$ such that
\[
    w_2 \cap \gamma b_{w_1} 
        ~ \neq ~ \emptyset \text{~~or~~} 
    w_1 \cap \gamma b_{w_2} 
        ~ \neq ~ \emptyset, 
\]
where the smaller cell was not generated by WSPs only involving the points of an inner cluster with respect to the second, we have by symmetry that
\[
    \frac{1}{\Delta} 
        ~ \leq ~ \frac{\diam(w)}{\diam(w')} 
        ~ \leq ~ \Delta. 
\]
For example, by setting $c_2 = 4$, $\gamma = 2$, $\alpha = 5$, and $c_1 = 15$, we have $\Delta < 41$.

It remains to account for blending between AVD cells and the smaller cells generated by points within their inner clusters. Fix an AVD cell $w$ with an inner cluster within the ball $b'_w$, and let $w'$ denote an AVD cell generated by a WSP with points in $b'_w$. The blending region for $w$ is chosen to keep the number of such cells small. Instead of blending over the entire $\gamma b_w$, we use a composite mollifier that effectively nullifies the contribution of $w$ within a suitable expansion of $\varrho b'_w$ with $\varrho \geq \beta$ (see Figure~\ref{inner_cluseter_blending.fig}). Hence, for AVD cells lying completely within $\varrho b'_w$, they need not enumerate over cells like $w$. Similarly, if the blending ball for $w'$ lies completely within $\varrho b'_w$, then it cannot overlap $w$ and no query point in $w$ needs to blend in the approximations computed for $w'$.

\begin{figure}[htbp]
  \centerline{\includegraphics[scale=0.40]{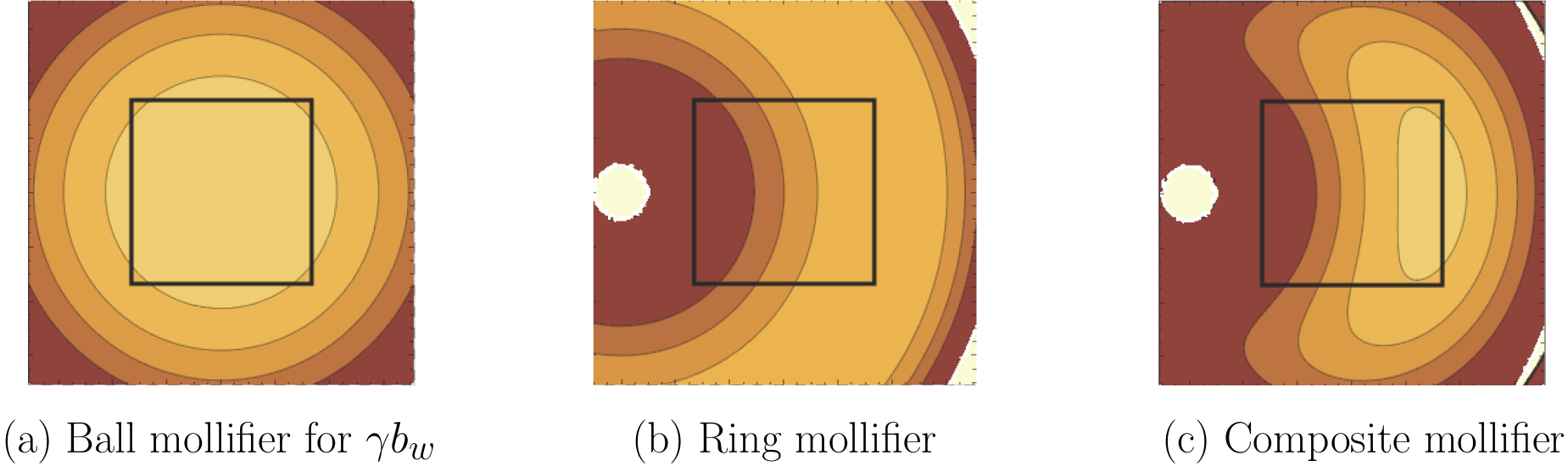}}
  \caption{Defining the blending region for an AVD cell $w$.}
  \label{inner_cluseter_blending.fig}
\end{figure}

With the composite mollifier, there remain two cases to consider: (a) The blending ball for $w'$, $\gamma b_{w'}$, does not lie completely within $\varrho b'_w$, and (b) $w'$ itself does not lie completely within $\varrho b'_w$. In the first case, some query points in $w$ may need to blend in the approximations computed for $w'$, while in the second case it is the other way around. Note that Case(b) implies Case(a), but the blending required for query points in $w$ and $w'$ may be different.

\begin{figure}[htbp]
  \centerline{\includegraphics[scale=0.40]{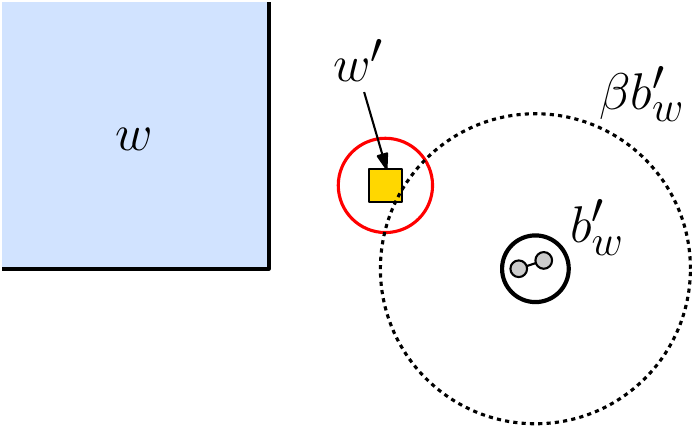}}
  \caption{Case(a): Bounding the diameter of AVD cells, generated by points in the inner cluster within $b'_w$, whose blending balls extend beyond the separation ball $\beta b'_w$.}
  \label{inner_blending.fig}
\end{figure}

For query points in $w$, we consider Case(a). Similar to the analysis above, we proceed to bound the possible values of the diameter of $w'$ to a small range of values. Only this time we do it relative to the radius of $b'_w$, which we denote by $r$. Let $z$ and $\ell$ the center and length, respectively, of the WSP within $b'_w$ which generated $w'$. It follows that $\ell \leq 2r$. Let $x$ denote the center of $w'$, and $y$ the point in $\gamma b_{w'}$ farthest from $z$, and observe that $z \in b'_w$, while $y \notin \beta b'_w$. It follows that for Case(a), we have:
\begin{align*}
    (\beta - 1) r 
        & ~ \leq ~ \dist(y, z) 
          ~ \leq ~ \dist(y, x) + \frac{\diam(w')}{2} + \dist(x, z) 
          ~ \leq ~ \gamma \frac{\ell}{2 c_1} + \frac{\ell}{2c_1} + c_2 \ell  \\
        & ~ =    ~ \frac{\gamma + 1 + 2 c_1 c_2}{2c_1}\ell,
\end{align*}
implying that $\ell \geq 2 c_1(\beta - 1)r/(\gamma + 1 + 2 c_1 c_2)$. For our choice of parameters, we get that $r/5 \leq \ell \leq 2 r$. Hence, $\diam(w)$, which is $\ell/c_1$ rounded down to a power of $2$, ranges over a set of constant size. By standard packing arguments, the number of all such boxes that overlap the ball $\beta b'_w$ is $O(1)$.

For query points in AVD cells like $w'$, we need to bound the number of larger cells such as $w$ whose blending regions may overlap $w'$. Towards this bound, we define $\varrho b'_w$ as the expansion of $b'_w$ that goes halfway between $\beta b'_w$ and $w$ (see Figure~\ref{inner_blending_2.fig}).

\begin{figure}[htbp]
  \centerline{\includegraphics[scale=0.40]{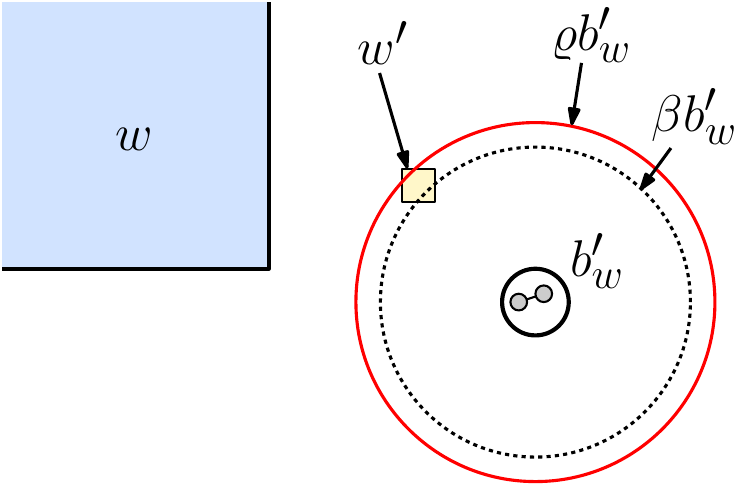}}
  \caption{Case(b): Bounding the number of AVD cells $w$ sharing an inner cluster in $b'w$ such that smaller AVD cells $w'$ at the periphery of $\beta b'_w$ may need to blend in $w$'s approximations.}
  \label{inner_blending_2.fig}
\end{figure}

Let $x$ denote the center of $b'_w$ and $y$ denote the point in $w'$ farthest from $x$. We also denote by $z$ and $\ell$ the center and length of the WSP that generated $w'$, and denote by $r$ the radius of $b'_w$; as before, we have $\ell \leq 2r$. Thus,
\[
    \dist(y, x) 
        ~ \leq ~ \diam(w') + \dist(w', z) + \dist(z, c) 
        ~ \leq ~ \frac{\ell}{c_1} + (\beta + 1) r + r 
        ~ \leq ~ \left( \frac{1}{c_1} + \beta + 2 \right) r.
\]
Hence, for $w'$ to extend outside $\varrho b'_w$, we must have:
\[
    \dist(w, \beta b'_w)/2 
        ~ \leq ~ \dist(y, c) - \beta r 
        ~ \leq ~ (1/c_1 + 2)r 
    \implies 
    \dist(w, c) 
        ~ \leq ~ (2/c_1+4+\beta)r.
\]
Towards a packing argument, we note that as $b'_w$ is an inner cluster for $w$, we have:
\[
    (\beta-1)r 
        ~ \leq ~ \alpha\cdot\diam(w)/2 
    \implies 
    \diam(w) 
        ~ \geq ~ \frac{2(\beta-1)}{\alpha}r.
\]
By our choice of parameters, we get that $\dist(w, c) \leq 7r$, while $\diam(w) \geq 2/5r$.

Hence, for any fixed AVD cell $w$, there can only be $O(1)$ cells, either smaller or larger, to consider for blending upon processing any query point $q \in w$. This completes the proof of Lemma~\ref{avd-lookup-bound.lem}.

\subsection{Compatibility of Shape Functions.} \label{sec:compatibility}

Recall the definitions of the shape functions employed in the Quadtree-EVD data structure per Eqs.~\eqref{eq:shape-b}--\eqref{eq:shape-r}.
\[
    \shape^{[b]}(x) 
        ~ = ~ \frac{1}{r^2}\left\|x - c^{[b]}\right\|^2, \quad 
    \shape^{[e]}(x) 
        ~ = ~ \left( x - c^{[e]} \right)^{\intercal}M\left( x - c^{[e]} \right), \quad 
    \shape^{[r]}(x) 
        ~ = ~ \dfrac{\left\|x - c^{[r]}\right\|^2 - (b^2+a^2)/2}{(b^2-a^2)/2}.
\]
In this section, we show that these are all compatible shape functions. (Recall Definition~\ref{def:compat-shape} from Section~\ref{sec:shape}.)

\subsubsection{Balls and Rings.}

In this section, we present proofs of Lemmas~\ref{quadtree-compatibility.lem} and~\ref{ring-compatibility.lem}, which show respectively that ball and run shape functions are both compatible.

\quadtreeCompatibility*

\begin{proof}
Recall that $r$ is the radius of the cell's minimum enclosing ball $b_w$, and $c^{[b]}$ is the center of this ball. By construction, a leaf-level quadtree cell $w$ is involved with a query point $q$ only if $q \in 2\cdot b_w$. Compatibility condition~(i) follows immediately.

To establish condition~(ii), we proceed to take derivatives as follows:
\begin{align*}
    \Gradient \shape^{[b]}(x) 
        & ~ = ~ \frac{2}{r^2}\left( x - c^{[b]} \right) 
         \implies \|\Gradient \shape^{[b]}(x)\| 
          ~ \leq ~ \frac{4}{r}, \text{~~for $x \in 2\cdot b_w$}, \\
    \Gradient^2 \shape^{[b]}(x) 
        & ~ = ~ \frac{2}{r^2} 
        \implies \|\Gradient^2 \shape^{[b]}(x)\| 
          ~ = ~ \frac{2}{r^2}.
\end{align*}
By Lemma~\ref{sep-props.lem}, $r  = \radius(b_w) \geq c \cdot d_P(x)$ for some constant $c$, which implies that $\|\Gradient \shape^{[b]}(x)\| = O(1/d_P(x))$ and $\|\Gradient^2 \shape^{[b]}(x)\| = O(1/d^2_P(x)) = O(1/(\eps \cdot d^2_P(x)))$, as desired.
\end{proof}

\ringCompatibility*

\begin{proof}
Recall that $a$ and $b$ are the ring's inner and outer radii, and $c^{[r]}$ is their common center. By construction of the revised EVD, a ring $R_i$ is only involved in queries with points $q$ only if $q \in R_i$. Compatibility condition~(i) follows immediately.

To establish condition~(ii), we proceed to take derivatives as follows:
\begin{align*}
    \Gradient \shape^{[r]}(x) 
        & ~ = ~ \frac{4}{b^2 - a^2}\left( x - c^{[r]} \right) 
          \implies \|\Gradient \shape^{[r]}(x)\| 
          ~ \leq ~ \frac{b}{b^2 - a^2}, \text{~~for $x \in R_i$}, \\
    \Gradient^2 \shape^{[r]}(x) 
        & ~ = ~ \frac{4}{b^2 - a^2} 
          \implies \|\Gradient^2 \shape^{[r]}(x)\| 
          ~ = ~ \frac{4}{b^2 - a^2}.
\end{align*}
By Lemma~\ref{sep-props.lem}, $\radius(b_w) \geq c \cdot d_P(x)$ for some constant $c$. This implies that $b > 2 a = \Omega(d_P(x))$, and hence $ \|\Gradient \shape^{[r]}(x)\| = O(1/d_P(x))$ and $\|\Gradient^2 \shape^{[r]}(x)\| = O(1/d^2_P(x)) = O(1/(\eps \cdot d^2_P(x)))$, as desired.
\end{proof}

Beyond establishing compatibility, we need an upper bound on the value of $\shape^{[r]}(q)$ when $q$ lies in the overlap of multiple rings $\{R_i\}$, as will be needed to establish a lower bound on $c_\Psi$ in Lemma~\ref{positive_Psi.lem}. Recall that the $i^{th}$ ring spans the interval $(3\cdot 2^{i-2}, 5\cdot 2^{i-1})\cdot r$, with $r$ the radius of the ball $b$ containing the inner set defining the rings $\{R_i\}$. Normalizing the scaling from $r$, and factoring out $3\cdot 2^{i-2}$,  we simplify the problem by reducing each pair of consecutive rings to the two staggered intervals $(1, 10/3)$ and $(2, 20/3)$. The relevant range, where the overlap happens, is the subinterval $(2, 10/3)$ where query points fall in both intervals. By substituting in the chosen functional form for $\shape^{[r]}$, we find that the maximum absolute value attained over the relevant subinterval is at most $0.72 < 3/4$, as shown in the plot of Figure~\ref{fa.fig}.

\begin{figure}[htbp]
  \centerline{\includegraphics[scale=0.50]{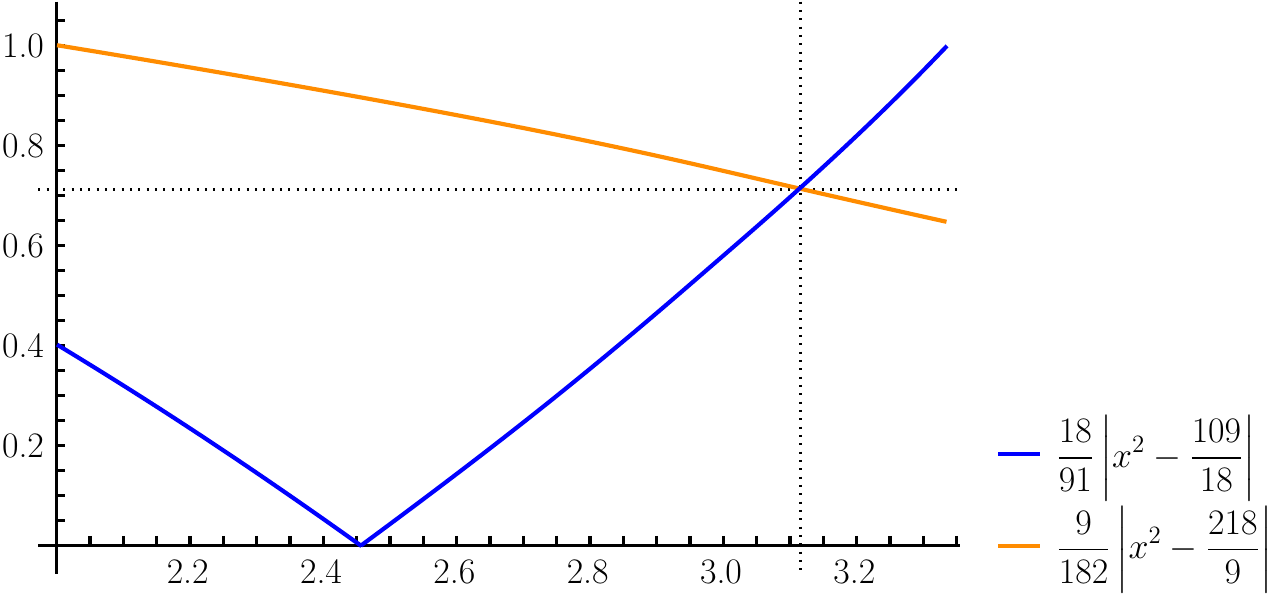}}
  \caption{Bounding the ring function using the staggered radii assignments.} \label{fa.fig}
\end{figure}

\subsubsection{Ellipsoids.}

In this section, we present a proof of Lemma~\ref{ellipsoid-compatibility.lem}, which shows that the ellipsoid shape functions are compatible. The proof relies on an analysis of the minimum eigenvalue of the matrices that define the ellipsoids. Given an ellipsoid defined by the equation $(x - c)^{\intercal}M(x - c)$, let $\alpha_{\min}(M)$ denote its smallest semi-axis and let $\lambda_{\max}(M)$ denote $M$'s maximum eigenvalue. We begin with the following technical result, which follows from basic linear algebra.

\begin{restatable}{lemma}{Mxc} \label{M_x_c.lem}
 Consider any distance-based Macbeath ellipsoid $E'$ centered at $c \in \RE^d$ with associated matrix $M$. Then for all $x \in E'$, $\|M(x - c)\| \leq \sqrt{d}/\alpha_{\min}(M)$.
\end{restatable}

For the EVD data structure we use, and the specific modifications described in Section~\ref{sec:blending}, we derive an explicit bound on $\lambda_{\max}(M_i)$ by bounding the smallest semi-axis of the ellipsoid defined by $M_i$, denoted by $\alpha_{\min}(M_i)$, which is equal to the square root of the reciprocal of $\lambda_{\max}(M_i)$.

\begin{restatable}{lemma}{alphaMinBound}\label{alpha_min_bound.lem}
Consider any distance-based Macbeath ellipsoid $E'$ with associated matrix $M$. Then there exists a constant $c$ such that for all $x \in E'$, $\alpha_{\min}(M) \geq c \SP \eps \cdot\DF(x))$.
\end{restatable}

\begin{proof}
We are given a query point $x$ which lies in the ellipsoid centered at $c^{[e]}$ with associated matrix $M$. Let $\ell$ denote the depth of $\cell(c^{[e]})$ in the EVD with the associated ellipsoid $E'_\ell(c^{[e]})$. By definition, $E'(c^{[e]}) = E^{1/2}_{\delta_\ell}(c^{[e]})$, or the $\frac{1}{2}$-scaled distance-based Macbeath ellipsoid at $c^{[e]}$ with respect to the expanded Voronoi cell $V_{\delta_\ell}(\nn(c^{[e]}))$, with $\delta_\ell = 2^\ell\cdot\sqrt{\eps}\cdot r$, where $r = \radius(b)$ and $b$ is a ball as defined in Section~\ref{sec:pou}.

In order to bound $\alpha_{\min}(M)$, which denotes the length of the shortest semi-axis of the ellipsoid $E'(c^{[e]})$, we go through a hypothetical ellipsoid $E_{\delta_\ell}(x)$ whose shape is easier to relate to $\DF(x)$. Then, we appeal to the properties of the Delone set defining the EVD to relate the shape of $E_{\delta_\ell}(x)$ to the shape of $E'(c^{[e]})$, and specifically $\alpha_{\min}(M)$.

Recall that $E_{\delta_\ell}(x) \subseteq M_{\delta_\ell}(x)$, where $M_{\delta_\ell}(x)$ is the distance-based Macbeath region at $x$ with respect to the expanded Voronoi cell $V_{\delta_\ell}(\nn(x))$. To establish a lower bound on $\alpha_{\min}(E_{\delta_\ell}(x))$, we assume that $x$ lies near the periphery of $V_{\delta_\ell}(\nn(x))$, for otherwise $\alpha_{\min}(E_{\delta_\ell}(x))$ can only be larger. Let $p = \nn(x)$ and $p' \in P$ denote any other point such that $x \in V_{\delta_\ell}(p')$.  Thus, the two shifted Voronoi bisectors defined by the pair $p$ and $p'$ support a pair of facets of the polytope $M_{\delta_\ell}(x)$. In order to bound $\alpha_{\min}(E_{\delta_\ell}(x))$, we establish a lower bound on the amount of shifting applied to the Voronoi bisector between $p$ and any such $p'$. Let $\alpha_\ell(p, p')$ denote the associated shift amount.

By the construction of the EVD, we have
\[
    r 
        ~ \leq ~ \DF(x) 
        ~ \leq ~ 2^\ell r.
\]
In addition, the construction of the EVD together with Lemma~\ref{exp-vor.lem} guarantee that $p'$ is an $\eps$-ANN of $x$, that is,
\[
    \|x - p'\| 
        ~ \leq ~ (1 + \eps) \DF(x),
\]
implying that $\|p - p'\| \leq 3\cdot\DF(x)$, by the triangle inequality and the fact that $\eps < 1$.

Substituting for $\delta_\ell$, and invoking the two bounds listed above, into the expression for $\alpha_\ell(p, p')$ from Section~\ref{sec:approx-voronoi}, we obtain
\[
    \alpha_\ell(p, p') 
        ~ =    ~ \frac{\delta_\ell^2}{2\|p - p'\|} 
        ~ \geq ~ \frac{2^{2\ell}\cdot\eps\cdot r^2}{6\cdot\DF(x)} 
        ~ \geq ~ \frac{\eps\cdot\DF^2(x)}{6\cdot\DF(x)} 
        ~ =    ~ \frac{\eps}{6}\cdot\DF(x),
\]
which implies that $\alpha_{\min}(E_{\delta_\ell}(x)) \geq \dfrac{\eps}{6} \cdot\DF(x)$.

To finish the proof, observe that $x \in E^{1/2}_{\delta_\ell}(c^{[e]}) \subseteq M^{1/2}_{\delta_\ell}(c^{[e]})$. By the expansion-containment property of Lemma~\ref{vor-exp-con.lem}, we have
\[
    x \in M^{1/2}_{\delta_\ell}(c^{[e]}) 
    ~ \implies ~ M^{1/2}_{\delta_\ell}(c^{[e]}) \cap M^{1/2}_{\delta_\ell}(x) \neq \emptyset 
    ~ \implies ~ M^{1/2}_{\delta_\ell}(x) \subseteq M^{7}_{\delta_\ell}(c^{[e]}).
\]
By applying John's theorem~\cite{Bal97} to the centrally-symmetric convex bodies $M^{1/2}_{\delta_\ell}(x)$ and $M^{7}_{\delta_\ell}(c^{[e]})$, we have
\[
    E^{1/2}_{\delta_\ell}(x) 
        ~\subseteq~ M^{1/2}_{\delta_\ell}(x) 
        ~\subseteq~ M^{7}_{\delta_\ell}(c^{[e]}) 
        ~\subseteq~ E^{7\sqrt{d}}_{\delta_\ell}(c^{[e]}),
\]
which implies that $E_{\delta_\ell}(x) \subseteq E^{14\sqrt{d}}_{\delta_\ell}(c^{[e]})$. By elementary results in convex geometry~\cite{BoV04}, it follows that
\[
    E_{\delta_\ell}(x) 
        ~ \subseteq ~ E^{14\sqrt{d}}_{\delta_\ell}(c^{[e]}) 
    ~ \implies ~
    \alpha_{\min}(E^{14\sqrt{d}}_{\delta_\ell}(c^{[e]})) 
        ~ \geq ~ \alpha_{\min}(E_{\delta_\ell}(x)) 
        ~ \geq ~ \frac{1}{6}\cdot\eps\cdot\DF(x).
\]
This shows that $\alpha_{\min}(M) = \alpha_{\min}(E^{1/2}_{\delta_\ell}(c^{[e]})) \geq \frac{1}{168\sqrt{d}}\cdot\eps\cdot\DF(x)$.
\end{proof}

Given this bound, we can now prove Lemma~\ref{ellipsoid-compatibility.lem}.

\ellipsoidCompatibility*

\begin{proof}
By construction, a Macbeath ellipsoid $E'(x)$ is involved with a query point $q$ only if $q \in E'(x)$. Compatibility condition~(i) follows immediately.

To establish condition~(ii), we proceed to take derivatives as follows:
\begin{align*}
    \Gradient \shape^{[e]}(x) 
        & ~ = ~ 2\cdot M(x - c^{[e]}) 
    ~ \implies ~
    \|\Gradient \shape^{[e]}(x)\| 
          ~ = ~ O(\sqrt{\lambda_{max}(M)}), \text{~~for $x \in E'(x)$}, \\
    \Gradient^2 \shape^{[e]}(x) 
        & ~ = ~ 2\cdot M 
    ~ \implies ~ 
    \|\Gradient^2 \shape^{[e]}(x)\| 
         ~ = ~ O(\lambda_{max}(M)).
\end{align*}
The proof follows from Lemmas~\ref{M_x_c.lem} and~\ref{alpha_min_bound.lem}.
\end{proof}


\pdfbookmark[1]{References}{s:ref}

\bibliographystyle{plainurl}
\bibliography{shortcuts,convex,references} 

\end{document}